\documentclass{article}
\usepackage[utf8]{inputenc}
\usepackage{amsmath}
\usepackage{amsthm}
\usepackage{amssymb}
\usepackage{fullpage}
\usepackage{authblk}
\usepackage{xcolor}
\usepackage{hyperref}
\usepackage{tikz}
\usepackage{subcaption}

\usetikzlibrary{arrows}
\usetikzlibrary{positioning, fit}
\usetikzlibrary{shapes}
\usetikzlibrary{decorations.pathmorphing}
\usetikzlibrary{calc,through,intersections}
\usetikzlibrary{patterns}
\usetikzlibrary{trees}
\usetikzlibrary{arrows.meta}
\tikzset{arc/.style = {->,> = latex', line width=.75pt}}

\title{Balancing graph Voronoi diagrams with one more vertex}
\author[1,2]{Guillaume Ducoffe}
\affil[1]{\small National Institute for Research and Development in Informatics, Romania}
\affil[2]{\small University of Bucharest, Romania}
\date{}

\newtheorem{lemma}{Lemma}[section]
\newtheorem{proposition}[lemma]{Proposition}
\newtheorem{theorem}[lemma]{Theorem}
\newtheorem{corollary}[lemma]{Corollary}

\newtheorem{problem}{Problem}
\newsavebox{\mybox}

\theoremstyle{definition}
\newtheorem{conjecture}{Conjecture}
\newtheorem{remark}{Remark}

\begin{document}

\maketitle

\begin{abstract}
    Let $G=(V,E)$ be a graph with unit-length edges and nonnegative costs assigned to its vertices. Being given a list of pairwise different vertices $S=(s_1,s_2,\ldots,s_p)$, the {\em prioritized Voronoi diagram} of $G$ with respect to $S$ is the partition of $G$ in $p$ subsets $V_1,V_2,\ldots,V_p$ so that, for every $i$ with $1 \leq i \leq p$, a vertex $v$ is in $V_i$ if and only if $s_i$ is a closest vertex to $v$ in $S$ and there is no closest vertex to $v$ in $S$ within the subset $\{s_1,s_2,\ldots,s_{i-1}\}$. For every $i$ with $1 \leq i \leq p$, the {\em load} of vertex $s_i$ equals the sum of the costs of all vertices in $V_i$. The load of $S$ equals the maximum load of a vertex in $S$. We study the problem of adding one more vertex $v$ at the end of $S$ in order to minimize the load. This problem occurs in the context of optimally locating a new service facility ({\it e.g.}, a school or a hospital) while taking into account already existing facilities, and with the goal of minimizing the maximum congestion at a site. There is a brute-force algorithm for solving this problem in ${\cal O}(nm)$ time on $n$-vertex $m$-edge graphs. We prove a matching time lower bound for the special case where $m=n^{1+o(1)}$ and $p=1$, assuming the so called Hitting Set Conjecture of Abboud et al. On the positive side, we present simple linear-time algorithms for this problem on cliques, paths and cycles, and almost linear-time algorithms for trees, proper interval graphs and (assuming $p$ to be a constant) bounded-treewidth graphs.
    
\end{abstract}

\section{Introduction}\label{sec:intro}

For undefined graph terminology, see~\cite{BoM08}. All graphs considered are finite, simple (they are loopless and have no multiple edges), connected and with unit-length edges.
Let $G=(V,E)$ be a graph.
Throughout the paper, let $n = |V|$ and $m = |E|$.
The distance between two vertices $u$ and $v$ is the minimum number of edges on a $uv$-path.
We denote this distance by $d_G(u,v)$, or simply $d(u,v)$ if $G$ is clear from the context.
Being given a vertex $v$ and a vertex subset $S=\{s_1,s_2,\ldots,s_p\}$, let also $d(v,S) = \min_{1 \leq i \leq p} d(v,s_i)$ denote the distance between $v$ and $S$.
A {\em Voronoi diagram} of $G$ with respect to $S$ is any partition $V_1,V_2,\ldots,V_p$ of $V$ such that, for every $i$ with $1 \leq i \leq p$, for every vertex $v$ with $v \in V_i$, $d(v,s_i) = d(v,S)$. 
This terminology originates from Computational Geometry, where the Voronoi diagram is a well-studied data structure~\cite{Aur91}.
An example of graph Voronoi diagram is given in Fig.~\ref{fig:voronoi}. 
We stress that in general, a Voronoi diagram is not uniquely defined, due to the possible existence of vertices $v$ such that $d(v,s_i) = d(v,s_j) = d(v,S)$, for some $i$ and $j$ with $1 \leq i < j \leq p$.

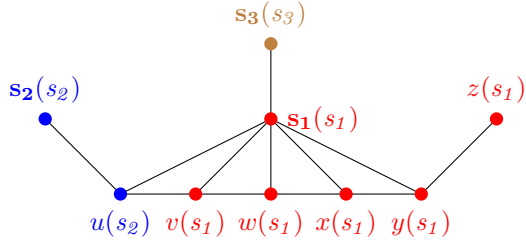
\begin{figure}[ht!]
    \centering
    \begin{tikzpicture}
        \draw (-3,4) -- (-2,3) -- (2,3) -- (3,4);
        \draw (-2,3) -- (0,4) -- (2,3);
        \draw (-1,3) -- (0,4) -- (1,3);
        \draw (0,5) -- (0,3);

        \node[circle,fill=brown,inner sep=0pt,minimum size=5pt,label=above:{$\color{brown}\mathbf{s_3} \mathit{(s_3)}$}] at (0,5) {};
        \node[circle,fill=red,inner sep=0pt,minimum size=5pt,label=right:{$\color{red}\mathbf{s_1} \mathit{(s_1)}$}] at (0,4) {};
    
        \node[circle,fill=blue,inner sep=0pt,minimum size=5pt,label=below:{$\color{blue}u \mathit{(s_2)}$}] at (-2,3) {};
        \node[circle,fill=red,inner sep=0pt,minimum size=5pt,label=below:{$\color{red}v \mathit{(s_1)}$}] at (-1,3) {};
        \node[circle,fill=red,inner sep=0pt,minimum size=5pt,label=below:{$\color{red}w \mathit{(s_1)}$}] at (0,3) {};
        \node[circle,fill=red,inner sep=0pt,minimum size=5pt,label=below:{$\color{red}x \mathit{(s_1)}$}] at (1,3) {};
        \node[circle,fill=red,inner sep=0pt,minimum size=5pt,label=below:{$\color{red}y \mathit{(s_1)}$}] at (2,3) {};
    
        \node[circle,fill=blue,inner sep=0pt,minimum size=5pt,label=above:{$\color{blue}\mathbf{s_2} \mathit{(s_2)}$}] at (-3,4) {};
        \node[circle,fill=red,inner sep=0pt,minimum size=5pt,label=above:{$\color{red}z \mathit{(s_1)}$}] at (3,4) {};
    \end{tikzpicture}
    \caption{A graph Voronoi diagram with respect to three vertices $s_1,s_2,s_3$. For each vertex of the graph, the corresponding vertex of $S$ is indicated in brackets. Since $d(s_1,u) = d(s_2,u)$, this diagram is not unique.}
    \label{fig:voronoi}
\end{figure}

To our best knowledge, Mehlorn was amongst the first to introduce the concept of graph Voronoi diagrams, which he used in the design of an approximation algorithm for the {\sc Steiner Tree} problem~\cite{Meh88}. Erwig further pursued the study of these objects, of which he presented some new algorithmic applications~\cite{Erw00}. Since then, graph Voronoi diagrams have appeared as a building block of even more algorithms~\cite{Cab18,CDW17,GKHM+18}. In~\cite{OSFS+08}, Okabe et al. introduced {\em generalized network Voronoi diagrams}, of which graph Voronoi diagrams are special cases. Other variations such as furthest color Voronoi diagrams were considered in~\cite{HKLS04}. We only consider graph Voronoi diagrams in what follows.

\medskip
Hakimi and Labb\'e introduced graph Voronoi diagrams, which they called Voronoi partitions, as a location-theoretic concept~\cite{HaL88}.
We build on their previous work and terminology.
Specifically, being given a vertex subset $S=\{s_1,s_2,\ldots,s_p\}$, and a Voronoi diagram $V_1,V_2,\ldots,V_p$ of $G$ with respect to $S$, we call {\em sites} the vertices in $S$, and we call {\em territories} the subsets $V_i$, with $1 \leq i \leq p$, of this partition. 
For concreteness, we may regard the sites in $S$ as locations of service facilities such as schools or hospitals.
With this interpretation in mind, every territory is the service area of some facility.
Several optimization criteria can be considered for these territories.
Being given some measure $\mu : V \times 2^V \to \mathbb{R}_{\geq 0}$ and a nonnegative integer $p$, the $\mu$-{\sc Voronoi} $p$-{\sc Center} problem is the problem of minimizing $\max_{1\leq i \leq p} \mu(s_i,V_i)$ amongst all possible graph Voronoi diagrams with respect to $p$ sites.
For instance, if $\mu(s_i,V_i) = \max_{v \in V_i} d(v,s_i)$, then this is the classical $p$-{\sc Center} problem.
%
Throughout the paper, let some function $\pi : V \to \mathbb{R}_{\geq 0}$ assign nonnegative costs to the vertices.
We call $p$-{\sc Balance} the $\mu_\pi$-{\sc Voronoi} $p$-{\sc Center} problem for the measure $\mu_\pi$ such that $\mu_\pi(s_i,V_i) = \sum_{v \in V_i}\pi(v)$.
For concreteness, we may regard $\mu_\pi(s_i,V_i)$, which we call the {\em load} of site $s_i$, as the population size in a service area.
We call our objective function $\max_{1 \leq i \leq p}\mu_\pi(s_i,V_i)$ the load of the Voronoi diagram.
The $p$-{\sc Balance} problem is NP-hard if $p$ is part of the input~\cite{HLS92}.
Interestingly, the NP-hardness reduction in~\cite{HLS92} outputs instances of the problem such that if there exists a graph Voronoi diagram with load at most some value $q$, then there always exist $p$ sites minimizing the load such that every vertex is at minimal distance of exactly one site.
This implies that the difficulty of this problem does not come from the non-uniqueness of a Voronoi diagram.

In practice, we must account for the location of existing service facilities in order to optimally locate a new one.
As a starter, let us consider the following problem (which is not quite yet what we study in the paper).
We are given a subset $S$ of sites, and the goal is to compute a new site $v$, with $v \notin S$, and a Voronoi diagram with respect to $S \cup \{v\}$ such that the load is minimized.
This problem is NP-hard even if $S$ is a singleton and $G$ is a complete graph.
Indeed, in this situation we can select any vertex $v \notin S$ as the new site, but the difficulty consists in computing an optimal graph Voronoi diagram.
The load of any such diagram must be at least $\pi(G)/2$, where $\pi(G)$ is the sum of the cost of all vertices.
Moreover, deciding whether there exists a solution with load $\pi(G)/2$ is equivalent to the well-known {\sc Partition} problem~\cite{GaJ79}.
In order to circumvent this negative result, we show that it is enough to force the Voronoi diagram to be uniquely defined.
More precisely, motivated by the practical need to minimally reorganize the diagram after inserting a new site, we follow Erwig's suggestion to prioritize the sites~\cite{Erw00}.
That is, $S=(s_1,s_2,\ldots,s_p)$ is a list, and every vertex $v$ must be assigned to its first closest site in the list.
Formally, for every $i$ with $1 \leq i \leq p$, we have $v \in V_i$ if and only if $d(v,s_i) = d(v,S)$ {\em and} $d(v,s_j) > d(v,S)$ for every $j$ with $1 \leq j < i$.
We call it the {\em prioritized Voronoi diagram}, or simply the Voronoi diagram of $G$ with respect to $S$, and we denote it in what follows by $\mathbf{Vor}(G,S)$.
Other conditions could be also considered.
For instance, one could add every vertex $v$ to the territories of {\em all} its closest sites in $S$, which is actually what Hakimi et al. did in~\cite{HLS92}.
However, having a vertex partition is important in some practical scenarios ({\it e.g.}, a student attends to only one school).
Furthermore, we stress that most results in the paper do not really depend on which condition we choose (a notable exception being the hardness result of Theorem~\ref{thm:hs-hard}).

\medskip
For every $i$ with $1 \leq i \leq p$, we denote by $\ell_\pi(s_i,S)$ the load of site $s_i$ in $\mathbf{Vor}(G,S)$.
The load of $\mathbf{Vor}(G,S)$ is denoted by $L_\pi(S)$. 
To our best knowledge, the following problem has not been considered before:

\begin{center}
	\fbox{
		\begin{minipage}{.95\linewidth}
			\begin{problem}[\textsc{Balanced Vertex}]\
				\label{prob:balanced-vertex} 
					\begin{description}
					\item[Input:] A graph $G=(V,E)$; A cost function $\pi : V \to \mathbb{R}_{\geq 0}$; A list $S=(s_1,s_2,\ldots,s_p)$ of pairwise different vertices.
					\item[Output:] A vertex $v \in V \setminus S$ such that $L_\pi(S+v)$ is minimized, where $S+v=(s_1,s_2,\ldots,s_p,v)$.
				\end{description}
			\end{problem}     
		\end{minipage}
	}
\end{center}

Closest to our work is the one-round Voronoi game on graphs~\cite{BBDS15}.
In the latter two-player game, Player $1$ first selects a subset $S$ of sites with some fixed cardinality.
Then, Player $2$ must select a $k$-set $S'$ of new sites, with $S' \cap S = \emptyset$, in order to {\em maximize} the cumulative load of all $k$ sites in $S'$ in the graph Voronoi diagram with respect to $S' + S$ (the new sites in $S'$ have higher priority than the former sites in $S$). 
By contrast, in our problem we aim at minimizing the load at any site, and we also account for the load of all former sites in $S$ and not just for the load of the new site.
Furthermore, the new site is given the least priority.

\paragraph{Our contributions.}
We initiate the complexity study of the {\sc Balanced Vertex} problem.

On general $n$-vertex $m$-edge graphs, the problem can be solved in ${\cal O}(nm)$ time by brute force, simply by considering each vertex $v \notin S$ (Theorem~\ref{thm:alg-gal}).
We prove a matching time lower bound assuming the so-called {\em Hitting Set Conjecture} (Theorem~\ref{thm:hs-hard}).
The latter was introduced in~\cite{AVW16} by Abboud et al. as a way to explain our lack of progress toward computing faster either the radius or the median of a graph. 
Since then, to our best knowledge, broader implications of the {Hitting Set Conjecture} have not been investigated.
Our work makes a step in this direction.
A better-studied conjecture, that is implied by the Hitting Set Conjecture and by the Strong Exponential Time Hypothesis, is the so called {\em Orthogonal Vectors Conjecture}~\cite{Wil05}.
We leave as an open problem whether we can prove a quadratic time lower bound for the {\sc Balanced Vertex} problem assuming this conjecture.

In the second and longest part of the paper, we break the quadratic barrier for {\sc Balanced Vertex} on several graph classes, with either some tree-like or path-like underlying structure. 
For that, we adapt tools for fast distance computation, one of them being the orthogonal range query framework of Cabello and Knauer~\cite{CaK09}, in order to compute the load $\ell_\pi(v,S+v)$ of every potential new site $v \notin S$.
If there are constantly many sites in $S$, then the latter is sufficient in order to solve {\sc Balanced Vertex}.
Indeed, by computing $\ell_{\pi'}(v,S'+v)$ for all sublists $S'$ of size one, and for some actualized cost functions $\pi'$, we can also compute the load $\ell_\pi(s,S+v)$ of every former site $s \in S$, and therefore we can compute $L_\pi(S+v)$. 
We detail this approach in Sec.~\ref{sec:tw}, where it is applied to bounded treewidth graphs (Theorem~\ref{thm:tw}).
However, solving {\sc Balanced Vertex} for an arbitrary (possibly nonconstant) number of sites looks more challenging.
For $n$-node trees, we combine a standard but rather intricate dynamic programming approach with centroid decomposition, so as to derive an ${\cal O}(n\log{n})$-time algorithm (Theorem~\ref{thm:tree}). The running time can be improved to ${\cal O}(n\log{|S|})$ (see Remark~\ref{rk:tree}).
Our second main contribution in this part is for proper interval graphs, for which we get an ${\cal O}(m+n\log{n})$-time algorithm (Theorem~\ref{thm:proper-int}).
This result is derived from prior works on distance labeling schemes for this class of graphs~\cite{GaP08}.

\paragraph{Organization of the paper.}
Basic results about the {\sc Balanced Vertex} problem, along with a polynomial-time algorithm for solving it and a matching conditional time lower bound, are presented in Sec.~\ref{sec:gal}.
Faster algorithms for special graph classes are presented in Sec.~\ref{sec:lin}.
We conclude the paper in Sec.~\ref{sec:ccl} with some open questions and perspectives.

\paragraph{Additional notations and terminology.}
Let $G=(V,E)$ be a graph.
The open neighbourhood of a vertex $v$ is denoted in what follows by $N_G(v) = \{u \in V \mid uv \in E\}$.
We denote by $N_G[v] = N_G(v) \cup \{v\}$ its closed neighbourhood.
The ball of center $v$ and radius $k$ is denoted by $N_G^k[v] = \{u \in V \mid d_G(u,v) \leq k\}$.
For every vertices $u$ and $v$, let $I_G(u,v) = \{ w \in V \mid d_G(u,v) = d_G(u,w) + d_G(w,v) \}$ denote the metric interval between $u$ and $v$.
Let $W_G(u,v) = \{ w \in V \mid d_G(u,w) < d_G(v,w) \}$ contain all vertices closer to $u$ than to $v$.
In the same way, for a vertex subset $S$ (a vertex list, resp.), let $N_G(S) = \bigcup_{s \in S}N_G(s) \setminus S$ and $N_G[S] = N_G(S) \cup S$ denote its open and closed neighbourhoods, respectively.
The subgraph induced by a subset $S$ is denoted by $G[S]$. 
For every vertex subset $S$ (vertex list, resp.), let $Proj_G(v,S) = \{ s \in S \mid d_G(v,s) = d_G(v,S) \}$ be the metric projection of a vertex $v$ on $S$.
We sometimes omit the subscript if $G$ is clear from the context.
A vertex list $S=(s_1,s_2,\ldots,s_p)$ such that $s_i \neq s_j$ for every $i,j$ with $1 \leq i < j \leq p$ is called in what follows a proper vertex list.
We denote the territories $V_1,V_2,\ldots,V_p$ of $\mathbf{Vor}(G,S)$ by $\mathbb{T}(s_1,S),\ldots,\mathbb{T}(s_p,S)$.
Finally, being given a cost function $\pi : V \to \mathbb{R}_{\geq 0}$, the cost of a vertex subset $X$ is defined as $\pi(X) = \sum_{x \in X} \pi(x)$.

\section{The general case}\label{sec:gal}

We here consider the {\sc Balanced Vertex} problem on general graphs.
In Sec.~\ref{sec:basic}, we prove a few simple properties of prioritized Voronoi diagrams.
We derive from the latter a quadratic-time algorithm, in the number of edges, for solving {\sc Balanced Vertex}.
This algorithm is presented in Sec.~\ref{sec:alg}.
In Sec.~\ref{sec:seth}, our quadratic running time is proved to be optimal assuming the Hitting Set Conjecture.

\subsection{Basic properties}\label{sec:basic}

We start the section with some easy results on the territories in a Voronoi diagram.
In particular, we prove the following relation between territories and metric intervals.

\begin{lemma}\label{lem:metric-interval}
Let $G=(V,E)$ be a graph and let $S=(s_1,s_2,\ldots,s_p)$ be a proper vertex list.
For every vertex $v \in \mathbb{T}(s_i,S)$, with $1 \leq i \leq p$, we have $I(v,s_i) \subseteq \mathbb{T}(s_i,S)$.
\end{lemma}
\begin{proof}
Suppose by contradiction the existence of some vertex $u \in I(v,s_i) \setminus \mathbb{T}(s_i,S)$.
Let $j$, with $1 \leq j \leq p$, such that $u \in \mathbb{T}(s_j,S)$.
In particular, $d(u,S) = d(u,s_j) \leq d(u,s_i)$.
Note that $d(v,s_j) \leq d(v,u) + d(u,s_j) \leq d(v,u) + d(u,s_i) = d(v,s_i) = d(v,S)$.
Therefore, we must have $d(u,s_j) = d(u,s_i) = d(u,S)$ and $d(v,s_i) = d(v,s_j) = d(v,S)$.
However, because $v \in \mathbb{T}(s_i,S)$, $j > i$, thus contradicting that $u \in \mathbb{T}(s_j,S)$.
\end{proof}

We immediately deduce from Lemma~\ref{lem:metric-interval} the following result about territories:

\begin{corollary}\label{cor:connected-cells}
Let $G=(V,E)$ be a graph and let $S=(s_1,s_2,\ldots,s_p)$ be a proper vertex list.
For every $i$ with $1 \leq i \leq p$, $\mathbb{T}(s_i,S)$ is a connected subset.
\end{corollary}

We often use in our proofs the following structural description of the last territory in a diagram.
Similar but more complicated descriptions of the other territories can be derived from this result.

\begin{lemma}\label{lem:struct-territory}
Let $G=(V,E)$ be a graph and let $S=(s_1,s_2,\ldots,s_p)$ be a proper vertex list.
For every vertex $v \notin S$, we have that $\mathbb{T}(v,S+v) = \bigcup_{1 \leq j \leq p} \left( W(v,s_j) \cap \mathbb{T}(s_j,S) \right)$
\end{lemma}
\begin{proof}
Let $j$ with $1 \leq j \leq p$ be fixed.
For every vertex $u \in \mathbb{T}(s_j,S)$, we have that $u \in \mathbb{T}(v,S+v)$ if and only if $d(u,v) < d(u,S) = d(u,s_j)$.
The vertices that satisfy this inequality are exactly those of $W(v,s_j)$. 
\end{proof}

\subsection{A quadratic-time algorithm}\label{sec:alg}

We present our quadratic-time algorithm for solving {\sc Balanced Vertex}.
It is based on the following observation that for every proper vertex list, we can compute the prioritized Voronoi diagram in linear time.

\begin{proposition}\label{prop:compute-voronoi}
Let $G=(V,E)$ be a graph and let $S=(s_1,s_2,\ldots,s_p)$ be a proper vertex list.
We can compute $\mathbf{Vor}(G,S)$ in ${\cal O}(n+m)$ time.
\end{proposition}
\begin{proof}
We run a BFS on $S$ ({\it i.e.}, we consider all vertices in $S$ as being one supervertex, on which we run a BFS).
In doing so, we order the vertices $v \in V \setminus S$ by nondecreasing distance to $S$.
Then, we compute $r(v)$, for every vertex $v$, so that:
$$r(v) = \begin{cases} v\text{, if }v \in S \\ 
\min\{r(u) \mid u \in N(v) \ \text{and} \ d(u,S) = d(v,S)-1\}\text{ otherwise.}\end{cases}$$
We prove by induction on $d(v,S)$ that $r(v)$ is the unique site of $S$ such that $v \in \mathbb{T}(r(v),S)$.
This is straightforward if $v \in S$.
Thus, from now on we assume that $d(v,S) \geq 1$.
For every $u \in N(v)$, if $d(u,S) < d(v,S)$, then by the induction hypothesis we have $r(u) \in Proj(v,S)$.
By Lemma~\ref{lem:metric-interval}, the unique site $s_i \in S$ such that $v \in \mathbb{T}(s_i,S)$ is contained in $\{r(u) \mid u \in N(v) \ \text{and} \ d(u,S) = d(v,S)-1\}$.
Therefore, $r(v) = s_i$.
\end{proof}

In order to solve {\sc Balanced Vertex}, it is sufficient to compute all prioritized Voronoi diagrams with one more site.
We summarize our observations as follows.

\begin{theorem}\label{thm:alg-gal}
We can solve {\sc Balanced Vertex} in ${\cal O}(nm)$ time.
\end{theorem}
\begin{proof}
For every vertex $v \in V \setminus S$, we compute $\mathbf{Vor}(G,S+v)$.
By Proposition~\ref{prop:compute-voronoi}, this computation takes ${\cal O}(n+m)$ time.
Furthermore, we can compute $L_\pi(S+v)$ in additional ${\cal O}(n)$ time if we scan each territory of $\mathbf{Vor}(G,S+v)$ sequentially.
The running time follows since there are at most $n$ vertices $v$ to be considered.
\end{proof}

\subsection{HS-hardness}\label{sec:seth}

Our main technical contribution in this section is a conditional time lower bound for {\sc Balanced Vertex}, which matches the quadratic running time of the simple algorithm presented in Theorem~\ref{thm:alg-gal}.
Our hardness result is based on the following hypothesis.

\begin{conjecture}[The Hitting Set Conjecture~\cite{AVW16}]\label{conj:hs}
There is no $\varepsilon > 0$ such that for all $c \geq 1$, there is an algorithm that given two lists $A,B$ of $n$ subsets of a universe $U$ of size at most $c \log{n}$, can decide in ${\cal O}(n^{2-\varepsilon})$ time if there is a set in the first list that intersects every set in the second list, {\it i.e.}, a ``hitting set''.
\end{conjecture}

In what follows, we call the triple $(A,B,U)$ a HS-instance.
Two HS-instances $(A,B,U)$ and $(A',B',U')$ are called equivalent if either both contain a hitting set or none of them does.

Insofar, the Hitting Set Conjecture has been understudied.
We next present reductions from arbitrary HS-instances to equivalent ones with some more structure (Lemmas~\ref{lem:hs-red-1} and~\ref{lem:hs-red-2}).
The running times of these reductions are roughly linear in the cumulative cardinalities of all subsets in the input, which (since there are $2n$ such subsets and the universe has size ${\cal O}(\log{n})$) is in ${\cal O}(n\log{n})$.

\begin{lemma}\label{lem:hs-red-1}
Every HS-instance $(A,B,U)$ can be reduced in ${\cal O}(n\log{n})$ time to some equivalent HS-instance $(A',B',U')$ with the following additional property: every set in $B'$ intersects at most half of the sets in $B'$. 
\end{lemma}
\begin{proof}
Let $(A_0,B_0,U_0)$ and $(A_1,B_1,U_1)$ be disjoint copies of $(A,B,U)$.
Let also $x_0,x_1 \notin U_0 \cup U_1$.
We set:
\begin{align*}
    A' &= \{a_0 \cup U_1 \mid a_0 \in A_0\} \cup \{a_1 \cup U_0 \mid a_1 \in A_1\}, \\
    B' &= \{b_0 \cup \{x_0\} \mid b_0 \in B_0\} \cup \{b_1 \cup \{x_1\} \mid b_1 \in B_1\}, \\
    U' &= U_0 \cup U_1 \cup \{x_0,x_1\}.
\end{align*}
The construction of $(A',B',U')$ takes ${\cal O}(n\log{n})$ time.

We can bipartition $B'$ in two sublists $B_0',B_1'$ such that, for every $i \in \{0,1\}$, all sets in $B_i'$ contain $x_i$ and are contained in $U_i \cup \{x_i\}$.
Since $U_0 \cup \{x_0\}$ and $U_1 \cup \{x_1\}$ are disjoint, every set of $B'$ intersects exactly half of the sets of $B'$.
Furthermore, if $a \in A$ is a hitting set for $(A,B,U)$, then let $a_0,a_1$ be its copies in $A_0,A_1$.
By construction, $a_0 \cup U_1$ and $a_1 \cup U_0$ are hitting sets for $(A',B',U')$.
Conversely, let $a' \in A'$ be a hitting set for $(A',B',U')$.
By symmetry, we may assume to have $a' = a_0 \cup U_1$, for some set $a_0 \in A_0$.
Since $x_0 \notin a'$ and, for every set $b_0 \in B_0$, $b_0 \cap U_1 = \emptyset$, we must have $a_0$ is a hitting set for $(A_0,B_0,U_0)$.
In particular, there is also a hitting set for $(A,B,U)$.
\end{proof}

We only need to apply the following reduction with $\alpha = 2, \ \beta = -1$ for proving Theorem~\ref{thm:hs-hard}.
However, the general case with $\alpha,\beta$ arbitrary might be useful for future works on the Hitting Set Conjecture.

\begin{lemma}\label{lem:hs-red-2}
Let $\alpha, \beta$ be arbitrary integers so that $\alpha \geq 2$.
Every HS-instance $(A,B,U)$ can be reduced in ${\cal O}(n\log{n})$ time to some equivalent HS-instance $(A',B,U')$ with the following additional properties: for some $t={\cal O}(|U|)$, the universe $U'$ has cardinality $\alpha \cdot t + \beta$, and every set in $A'$ has cardinality $t$. 
\end{lemma}
\begin{proof}
First, we reduce $(A,B,U)$ to some equivalent instance $(A_{\text{tmp}},B,U_{\text{tmp}})$ so that all sets in $A_{\text{tmp}}$ have equal cardinality.
Then, we reduce $(A_{\text{tmp}},B,U_{\text{tmp}})$ to our final HS-instance $(A',B,U')$.

In order to construct $(A_{\text{tmp}},B,U_{\text{tmp}})$, we start computing the extreme cardinalities $\Delta = \max\{|a| \mid a \in A\}$ and $\delta = \min\{ |a| \mid a \in A \}$ of sets in $A$.
Let $U_{\text{dummy}} = \{u_1,u_2,\ldots,u_{\Delta-\delta}\}$ be a universe of cardinality $\Delta-\delta$ that is disjoint from $U$.
We set $U_{\text{tmp}} = U \cup U_{\text{dummy}}$.
For every set $a \in A$, we add in $A_{\text{tmp}}$ the set $a \cup \{u_1,u_2,\ldots,u_{\Delta-|a|}\}$.

Let every set in $A_{\text{tmp}}$ have cardinality $\Delta$.
In order to construct $(A',B',U')$, there are two different cases to be considered.
\begin{itemize}
    \item \underline{Case $\alpha \cdot \Delta + \beta > |U_{\text{tmp}}|$}. We complete $U_{\text{tmp}}$ with $\alpha \cdot \Delta + \beta -|U_{\text{tmp}}|$ new elements in order to create $U'$. The two lists are unchanged ({\it i.e.}, we set $A' = A_{\text{tmp}}$, $B' = B$).
    \item \underline{Case $\alpha \cdot \Delta + \beta < |U_{\text{tmp}}|$}. Let us write $|U_{\text{tmp}}|-\alpha \cdot \Delta - \beta = q \cdot (\alpha-1) + r$ for some unique $q \geq 0$ and $r \in \{0,1,\ldots,\alpha-2\}$. Let $X,Y$ be disjoint universes of respective cardinalities $\alpha-1-r$ and $q+1$, both disjoint from $U_{\text{tmp}}$. We set $U' = U_{\text{tmp}} \cup X \cup Y$, $A' = \{ a \cup Y \mid a \in A_{\text{tmp}}\}$, $B'=B$. Doing so, every set in $A'$ has cardinality equal to $t = \Delta+q+1$, while the universe $U'$ has cardinality:
    \begin{align*}
        |U_{\text{tmp}}| + |X| + |Y| &= \alpha \cdot \Delta + \beta + (|U_{\text{tmp}}|-\alpha \cdot \Delta - \beta) + (q + 1) + (\alpha - 1 - r) \\
        &= \alpha \cdot \Delta + \beta + (q \cdot (\alpha-1) + r) + q + \alpha - r \\
        &= \alpha \cdot \Delta + \beta + q \cdot \alpha + \alpha \\
        &= \alpha \cdot (\Delta+q+1) + \beta \\
        &= \alpha \cdot t + \beta
    \end{align*}
\end{itemize}
In both cases, our construction ensures that every set in $A'$ has cardinality $\frac{|U'|-\beta}{\alpha}$.
\end{proof}

We are now ready to combine both reductions in order to prove our main result in this part.

\begin{theorem}\label{thm:hs-hard}
Assuming Conjecture~\ref{conj:hs}, there is no $\varepsilon > 0$ such that for all $c \geq 1$, there is an algorithm that solves {\sc Balanced Vertex} in ${\cal O}(n^{2-\varepsilon})$ time on the $n$-vertex graphs with at most $cn\log{n}$ edges.
\end{theorem}
\begin{proof}
Let $(A,B,U)$ be any HS-instance with the following properties:
\begin{enumerate}
    \item every set of $B$ intersects at most half of the sets in $B$;
    \item every set of $A$ has cardinality $\frac{|U|+1}{2}$.
\end{enumerate}
We stress that every HS-instance can be reduced to an equivalent one with these properties by applying the reductions of Lemma~\ref{lem:hs-red-1} and~\ref{lem:hs-red-2}, in this order.
Intuitively, we identify the subsets of $A \cup B$ and the elements of $U$ with vertices in some graph $G$ (to be defined later).
We further define a unique site $s$, which is close to every vertex in $A$ but distant to every vertex in $B$.
Therefore, in order to minimize the load, one should pick a new site which is closer than $s$ to every vertex in $B$.
At first glance, one may think about picking some new site in $U$. However, the vertices in $U$ are closer than $s$ to far too many vertices, thus making the load increase.
The first property above ensures that no vertex of $B$ should be picked either.
Therefore, we are left picking (if any) a vertex of $A$ that is close to every vertex of $B$, {\it a.k.a.}, a hitting set.
Note that in general, the exact load should also depend on the size of $U$ and the cardinality of the hitting set.
The second property above ensures that it is not the case for our reduction.

For convenience, in what follows let $t=\frac{|U|+1}{2}$ denote the cardinality of every set in $A$.
We may further assume that $t >2$ (otherwise, we can detect a hitting set in ${\cal O}(n)$ time because there are ${\cal O}(1)$ distinct sets in $A$).
We construct a graph $G=(V,E)$ with $2n+|U|+3 = 2(n+t+1)$ vertices, as follows:
\begin{itemize}
    \item $V = A \cup B \cup U \cup \{s,x,y\}$;
    \item $U \cup \{x,y\}$ is a clique;
    \item $N_G(s) = \{x,y\}$;
    \item for every set $a \in A$, $N_G(a) = \{y\} \cup \{ u \in U \mid u \in a \}$;
    \item for every set $b \in B$, $N_G(b) = \{u \in U \mid u \in b\}$.
\end{itemize}

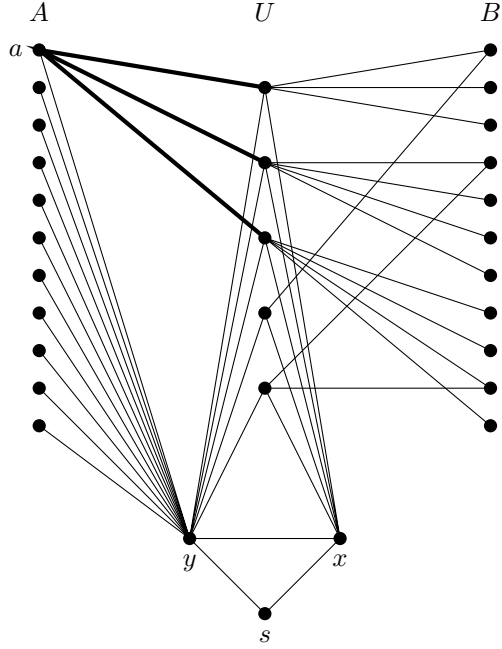
\begin{figure}[!ht]
    \centering
    \begin{tikzpicture}
        
        \node[circle,fill=black,inner sep=0pt,minimum size=5pt,label=above:{}] at (0,2) {};
        \node[circle,fill=black,inner sep=0pt,minimum size=5pt,label=above:{}] at (0,1) {};
        \node[circle,fill=black,inner sep=0pt,minimum size=5pt,label=above:{}] at (0,0) {};
        \node[circle,fill=black,inner sep=0pt,minimum size=5pt,label=above:{}] at (0,-1) {};
        \node[circle,fill=black,inner sep=0pt,minimum size=5pt,label=above:{}] at (0,-2) {};
        
        \node[circle,fill=black,inner sep=0pt,minimum size=5pt,label=left:{$a$}] at (-3,2.5) {};
        \node[circle,fill=black,inner sep=0pt,minimum size=5pt,label=above:{}] at (-3,2) {};
        \node[circle,fill=black,inner sep=0pt,minimum size=5pt,label=above:{}] at (-3,1.5) {};
        \node[circle,fill=black,inner sep=0pt,minimum size=5pt,label=above:{}] at (-3,1) {};
        \node[circle,fill=black,inner sep=0pt,minimum size=5pt,label=above:{}] at (-3,.5) {};
        \node[circle,fill=black,inner sep=0pt,minimum size=5pt,label=above:{}] at (-3,0) {};
        \node[circle,fill=black,inner sep=0pt,minimum size=5pt,label=above:{}] at (-3,-.5) {};
        \node[circle,fill=black,inner sep=0pt,minimum size=5pt,label=above:{}] at (-3,-1) {};
        \node[circle,fill=black,inner sep=0pt,minimum size=5pt,label=above:{}] at (-3,-1.5) {};
        \node[circle,fill=black,inner sep=0pt,minimum size=5pt,label=above:{}] at (-3,-2) {};
        \node[circle,fill=black,inner sep=0pt,minimum size=5pt,label=above:{}] at (-3,-2.5) {};
        
        \node[circle,fill=black,inner sep=0pt,minimum size=5pt,label=above:{}] at (3,2.5) {};
        \node[circle,fill=black,inner sep=0pt,minimum size=5pt,label=above:{}] at (3,2) {};
        \node[circle,fill=black,inner sep=0pt,minimum size=5pt,label=above:{}] at (3,1.5) {};
        \node[circle,fill=black,inner sep=0pt,minimum size=5pt,label=above:{}] at (3,1) {};
        \node[circle,fill=black,inner sep=0pt,minimum size=5pt,label=above:{}] at (3,.5) {};
        \node[circle,fill=black,inner sep=0pt,minimum size=5pt,label=above:{}] at (3,0) {};
        \node[circle,fill=black,inner sep=0pt,minimum size=5pt,label=above:{}] at (3,-.5) {};
        \node[circle,fill=black,inner sep=0pt,minimum size=5pt,label=above:{}] at (3,-1) {};
        \node[circle,fill=black,inner sep=0pt,minimum size=5pt,label=above:{}] at (3,-1.5) {};
        \node[circle,fill=black,inner sep=0pt,minimum size=5pt,label=above:{}] at (3,-2) {};
        \node[circle,fill=black,inner sep=0pt,minimum size=5pt,label=above:{}] at (3,-2.5) {};

        \node[circle,fill=black,inner sep=0pt,minimum size=5pt,label=below:{$s$}] at (0,-5) {};
        \node[circle,fill=black,inner sep=0pt,minimum size=5pt,label=below:{$x$}] at (1,-4) {};
        \node[circle,fill=black,inner sep=0pt,minimum size=5pt,label=below:{$y$}] at (-1,-4) {};

        \draw (1,-4) -- (-1,-4);

        \node at (-3,3) {$A$};
        \node at (0,3) {$U$};
        \node at (3,3) {$B$};
        
        \draw (1,-4) -- (0,-5) -- (-1,-4);
        
        \draw (-3,2.5) -- (-1,-4) -- (-3,2);
        \draw (-3,1.5) -- (-1,-4) -- (-3,1);
        \draw (-3,.5) -- (-1,-4) -- (-3,0);
        \draw (-3,-.5) -- (-1,-4) -- (-3,-1);
        \draw (-3,-1.5) -- (-1,-4) -- (-3,-2);
        \draw (-3,-2.5) -- (-1,-4);
        
        \draw (0,2) -- (-1,-4) -- (0,1); \draw (0,0) -- (-1,-4) -- (0,-1); \draw (0,-2) -- (-1,-4);
        \draw (0,2) -- (1,-4) -- (0,1); \draw (0,0) -- (1,-4) -- (0,-1); \draw (0,-2) -- (1,-4);
        
        \draw[ultra thick] (0,2) -- (-3,2.5) -- (0,1); \draw[ultra thick] (-3,2.5) -- (0,0);
        
        \draw (3,2.5) -- (0,2) -- (3,2); \draw (3,1.5) -- (0,2);
        \draw (3,1) -- (0,1) -- (3,.5); \draw (3,0) -- (0,1) -- (3,-.5);
        \draw (3,-1) -- (0,0) -- (3,-1.5); \draw (3,-2) -- (0,0) -- (3,-2.5);
        \draw (3,2.5) -- (0,-1);
        \draw (3,1) -- (0,-2) -- (3,-2);
        
    \end{tikzpicture}
    \caption{The reduction of Theorem~\ref{thm:hs-hard}. For clarity of the picture, we omitted the edges between vertices of $U$ and the edges between $U$ and $A \setminus \{a\}$, with $a$ some hitting set.}
    \label{fig:red}
\end{figure}

See Fig.~\ref{fig:red} for an illustration.
Let $S=(s)$ and let $\pi : V \to \mathbb{R}_{\geq 0}$ be the all-one function ({\it i.e.}, $\pi(v) = 1$ for every vertex $v$).
Being given a vertex $v \neq s$, let $S_v = (s,v)$.
By Lemma~\ref{lem:struct-territory}, $\mathbb{T}(v,S_v) = W_G(v,s)$.
We now compute the load of $S_v$ for every vertex $v \neq s$:
\begin{itemize}
    \item {\it Case $v = x$}. We have $W(x,s) = U \cup \{x\} \cup B$. Then, $L_\pi(S_x) = \ell_\pi(x,S_x) = n+1+|U| = n+2t$.
    \item {\it Case $v = y$}. We have $W(y,s) = V \setminus \{s,x\}$. Then, $L_\pi(S_y) = \ell_\pi(y,S_y) = 2n+|U|+1 = 2(n+t)$.
    \item {\it Case $v \in U$}. We have $U \cup B \subseteq W(v,s)$. Then, $L_\pi(S_v) = \ell_\pi(v,S_v) \geq n+|U|=n+2t-1$.
    \item {\it Case $v \in B$}. Since every vertex of $A$ is at distance at least two from vertex $v$, $A \cap W(v,s) = \emptyset$. Furthermore, $|W(v,s) \cap B| \leq |B|/2$. We also know that $s,x,y \notin W(v,s)$. As a result, we have $L_\pi(S_v) = \ell_\pi(s,S_v) \geq |A|+|B|/2+3 = 3n/2+3$.
    \item {\it Case $v \in A$}. We have $W(v,s) = \{v\} \cup \left(N_G(v) \cap U\right) \cup \left(N_G^2[v] \cap B\right)$. Let us write $k = |N_G^2[v] \cap B|$. In this situation, $\ell_\pi(v,S_v) = k+t+1$ while $\ell_\pi(s,S_v) = 2n-k +t + 1$. In particular, $L_\pi(S_v)$ is minimal if and only if $k=n$, that is if and only if $v$ is a hitting set.
\end{itemize}
As a result, the existence of some vertex $v$ such that $L_{\pi}(S_v) \leq n+t+1$ is equivalent to that of a hitting set for $(A,B,U)$.
\end{proof}

\begin{remark}\label{rk:hardness}
The hardness result of Theorem~\ref{thm:hs-hard} holds for unit costs and $|S|=1$.
\end{remark}

\section{Almost linear-time algorithms}\label{sec:lin}

In this section, we break the quadratic barrier for {\sc Balanced Vertex} on various graph classes.
We obtain a linear-time algorithm only for simple topologies, see Sec.~\ref{sec:easy}.
Nevertheless, we achieve an almost linear-time computation on trees (Sec.~\ref{sec:tree}), proper interval graphs (Sec.~\ref{sec:interval}) and bounded-treewidth graphs with a constant number of sites (Sec.~\ref{sec:tw}).

\subsection{Simple cases}\label{sec:easy}

First, we present a complete dichotomy of the parameterized complexity of {\sc Balanced Vertex} if the parameter is the diameter of the graph.
The graphs of unit diameter are exactly the complete graphs.
Being given a proper vertex list $S$ in such graph, every vertex of $V \setminus S$ must belong to the territory of the site in $S$ with the highest priority. 
We deduce from this observation a simple linear-time algorithm for {\sc Balanced Vertex} on complete graphs.

\begin{lemma}\label{lem:clique}
We can solve {\sc Balanced Vertex} in ${\cal O}(n)$ time if $G$ is a clique.
\end{lemma}
\begin{proof}
Let $\Lambda = \pi(G) - \sum_{i=2}^p \pi(s_i)$.
It suffices to output a vertex $v \notin S$ such that $\max\{\pi(v),\Lambda-\pi(v)\}$ is minimized.
Indeed, since $G$ is a clique we have the following for every vertex $v \notin S$: $\mathbb{T}(s_1,S+v) = V \setminus \{s_2,s_3,\ldots,s_p,v\} = \mathbb{T}(s_1,S) \setminus \{v\}$; $\mathbb{T}(s_i,S+v) = \{s_i\} = \mathbb{T}(s_i,S)$ for every $i$ with $2 \leq i \leq p$; $\mathbb{T}(v,S+v) = \{v\}$. 
In particular, we have $L_\pi(S+v) = \max\{\pi(s_i) \mid 2 \leq i \leq p\} \cup \{\pi(v),\Lambda-\pi(v)\}$.
\end{proof}

We now generalize the simple idea behind Lemma~\ref{lem:clique} to the graphs of diameter two.

\begin{lemma}\label{lem:diam-two}
We can solve {\sc Balanced Vertex} in ${\cal O}(n+m)$ time if $G$ has diameter two.
\end{lemma}
\begin{proof}
We compute $\mathbf{Vor}(G,S)$, which by Proposition~\ref{prop:compute-voronoi} we can do in ${\cal O}(n+m)$ time.
Doing so, we can compute in ${\cal O}(n)$ time the values $\ell_\pi(s_i,S)$, for every $i$ with $1 \leq i \leq p$.
We can also compute in ${\cal O}(n)$ time, for every $i$ with $1 \leq i \leq p$, the value $\mu_i$ such that: $$\mu_i = \begin{cases}
\max\{ \ell_\pi(s_2,S), \ell_\pi(s_3,S), \ldots, \ell_\pi(s_p,S) \} \ \text{if} \ i=1 \\
\max\{ \ell_\pi(s_j,S) \mid j \in \{2,3,\ldots,p\} \setminus \{i\} \} \ \text{otherwise.} 
\end{cases}$$

For every vertex $v \notin S$, let $s_i$ be the unique site of $S$ such that $v \in \mathbb{T}(s_i,S)$. Note that we can compute $s_i$ from $\mathbf{Vor}(G,S)$ in ${\cal O}(1)$ time.
We compute the value $\lambda(v)$ such that:
$$\lambda(v) = \begin{cases}
\max\left\{\pi\left(N(v) \setminus N[S]\right)+\pi(v), \ell_\pi(s_1,S) - \pi\left(N(v) \setminus N[S]\right) - \pi(v), \mu_1\right\} \ \text{if} \ i=1 \\
\max\left\{ \pi\left(N(v) \setminus N[S]\right)+\pi(v), \ell_\pi(s_1,S) - \pi\left(N(v) \setminus N[S]\right), \ell_\pi(s_i,S) - \pi(v), \mu_i \right\} \ \text{otherwise.}
\end{cases}$$
Since we precomputed all values $\mu_1,\mu_2,\ldots,\mu_p$, we can compute the values $\lambda(v)$, for every $v \notin S$, in ${\cal O}(n+m)$ time, simply by scanning the neighbourhoods of every vertex.
Finally, we output a vertex $v \notin S$ such that $\lambda(v)$ is minimized.
Indeed, we claim that $\lambda(v) = L_\pi(S+v)$.
This is because we assume the diameter of $G$ to be two, and therefore $\mathbb{T}(s_j,S) \subseteq N[s_j]$ for every $j$ with $2 \leq j \leq p$.
In particular, we have: $$W_G(v,s_j) \cap \mathbb{T}(s_j,S) = \begin{cases}
\{v\} \ \text{if} \ v \in \mathbb{T}(s_j,S) \\
\emptyset \ \text{otherwise.}
\end{cases}$$
In the former case, $\ell_\pi(s_j,S+v) = \ell_\pi(s_j,S) - \pi(v)$, while in the latter case $\ell_\pi(s_j,S+v) = \ell_\pi(s_j,S)$.
Furthermore, $\mathbb{T}(v,S+v) = \{v\} \cup \left(N(v) \setminus N[S] \right)$, and so, $\ell_\pi(v,S+v) = \pi\left(N(v) \setminus N[S]\right) + \pi(v)$.
Finally,
$$W_G(v,s_1) \cap \mathbb{T}(s_1,S) = \begin{cases}
\{v\} \cup \left(N(v) \setminus N[S] \right) \ \text{if} \ v \in \mathbb{T}(s_1,S) \\
N(v) \setminus N[S] \ \text{otherwise.}
\end{cases}$$
Altogether combined, we obtain as claimed that $\lambda(v) = L_\pi(S+v)$.
\end{proof}

We cannot extend Lemmas~\ref{lem:clique} and~\ref{lem:diam-two} any further because the graph constructed in the proof of Theorem~\ref{thm:hs-hard} has diameter three.
Therefore, we can solve {\sc Balanced Vertex} in linear time if the diameter is at most two, while we cannot solve this problem in truly subquadratic time if the diameter is at least three.

\medskip
Next, we consider paths and cycles.
The problem becomes local, in the following sense.
Let us subdivide a path $G$ in maximal subpaths whose every end is either a site of $S$ or one of the two ends of $G$.
See Fig.~\ref{fig:path} for an illustration.
Doing so, we can prove that there are at most two sites of $S$ whose territory may be modified upon inserting a new site $v$, namely, those at the ends of the maximal subpath containing $v$. 
In particular, we can compute the load $L_\pi(S+v)$, for every node $v \notin S$, by applying to every maximal subpath a classical partial sum trick for vector problems.

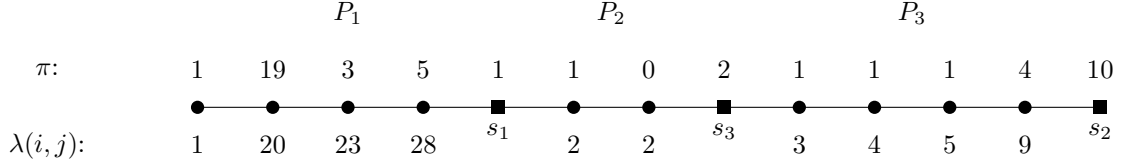
\begin{figure}[!ht]
    \centering
    \begin{tikzpicture}
        
        \node at (-8,.5) {$\pi$:}; \node at (-8,-.5) {$\lambda(i,j)$:};
        
        \node at (-6,.5) {$1$};
        \node at (-5,.5) {$19$};
        \node at (-4,.5) {$3$};
        \node at (-3,.5) {$5$};
        \node at (-2,.5) {$1$};
        \node at (-1,.5) {$1$};
        \node at (0,.5) {$0$};
        \node at (1,.5) {$2$};
        \node at (2,.5) {$1$};
        \node at (3,.5) {$1$};
        \node at (4,.5) {$1$};
        \node at (5,.5) {$4$};
        \node at (6,.5) {$10$};
        
        \node at (-6,-.5) {$1$};
        \node at (-5,-.5) {$20$};
        \node at (-4,-.5) {$23$};
        \node at (-3,-.5) {$28$};
        \node at (-1,-.5) {$2$};
        \node at (0,-.5) {$2$};
        \node at (2,-.5) {$3$};
        \node at (3,-.5) {$4$};
        \node at (4,-.5) {$5$};
        \node at (5,-.5) {$9$};
        
        \node[circle,fill=black,inner sep=0pt,minimum size=5pt,label=below:{}] at (-6,0) {};
        \node[circle,fill=black,inner sep=0pt,minimum size=5pt,label=below:{}] at (-5,0) {};
        \node[circle,fill=black,inner sep=0pt,minimum size=5pt,label=below:{}] at (-4,0) {};
        \node[circle,fill=black,inner sep=0pt,minimum size=5pt,label=below:{}] at (-3,0) {};
        \node[rectangle,fill=black,inner sep=0pt,minimum size=5pt,label=below:{$s_1$}] at (-2,0) {};
        \node[circle,fill=black,inner sep=0pt,minimum size=5pt,label=below:{}] at (-1,0) {};
        \node[circle,fill=black,inner sep=0pt,minimum size=5pt,label=below:{}] at (0,0) {};
        \node[rectangle,fill=black,inner sep=0pt,minimum size=5pt,label=below:{$s_3$}] at (1,0) {};
        \node[circle,fill=black,inner sep=0pt,minimum size=5pt,label=below:{}] at (2,0) {};
        \node[circle,fill=black,inner sep=0pt,minimum size=5pt,label=below:{}] at (3,0) {};
        \node[circle,fill=black,inner sep=0pt,minimum size=5pt,label=below:{}] at (4,0) {};
        \node[circle,fill=black,inner sep=0pt,minimum size=5pt,label=below:{}] at (5,0) {};
        \node[rectangle,fill=black,inner sep=0pt,minimum size=5pt,label=below:{$s_2$}] at (6,0) {};
        
        \draw (-6,0) -- (6,0);
        
        \node at (-4,1.25) {$P_1$};
        \node at (-.5,1.25) {$P_2$};
        \node at (3.5,1.25) {$P_3$};
        
    \end{tikzpicture}
    \caption{Subdivision of a path $G$ in maximal subpaths with respect to the sites in $S$.}
    \label{fig:path}
\end{figure}

\begin{lemma}\label{lem:path}
We can solve {\sc Balanced Vertex} in ${\cal O}(n)$ time if $G$ is a path.
\end{lemma}
\begin{proof}
We first describe our algorithm.
See Fig.~\ref{fig:path} for an illustration.
\begin{enumerate}
    \item We compute $\mathbf{Vor}(G,S)$ and the loads $\ell_\pi(s_j,S)$, for every $j$ with $1 \leq j \leq p$.
    \item Let $C_1,C_2,\ldots,C_q$ be the connected components of $G \setminus S$.
    Note that we have $q \leq p+1$.
    For every $i$ with $1 \leq i \leq q$, let $S_i = N(C_i)$, and let $P_i$ be the subpath induced by $N[C_i] = C_i \cup S_i$.
    We stress that $|S_i|=2$, except maybe if $i \in \{1,q\}$ for which we may have $|S_i| = 1$.
    \item We consider each subpath $P_i$, with $1 \leq i \leq q$, sequentially.
    In what follows, let $P_i=(w_0^i,w_1^i,\ldots,w_{t_i-1}^i)$.
    \begin{itemize}
        \item We compute $\Lambda_i = \max_{s \in S \setminus S_i} \ell_\pi(s,S)$.
        \item We compute $\lambda(i,j) = \sum_{k=0}^j \pi(w_k^i)$, for every $k$ with $0 \leq k < t_i$.
    \end{itemize}
    Then, there are two different cases to be considered.
    \begin{enumerate}
        \item {\it Case $|S_i| = 1$}. Either $i=1, \ w_0^1 \notin S_1$ or $i=q, \ w_{t_q-1}^q \notin S_q$. 
        We observe that both subcases are symmetric up to reverting the path $G$.
        Therefore, we assume without loss of generality that $i=1, \ w_0^1 \notin S_1$.
        In this situation, $w_0^1$ is the leftmost end of path $G$, and $S_1 = \{w_{t_1-1}^1\}$.
        For every $j$ with $0 \leq j < t_1-1$, let $j_R = \left\lceil \frac{t_1-1+j}{2}\right\rceil - 1$. We compute:
        $$\Lambda(v) = \max\{\Lambda_1,\lambda(1,j_R),\ell_\pi(w_{t_1-1}^1,S)-\lambda(1,j_R)\}$$
        \item {\it Case $|S_i| = 2$}. We first compute $j^i_{\lim} = \max\{ j \mid 0 \leq j < \ell_i, \ w_j^i \in \mathbb{T}(w_0^i,S) \}$.
        Then, we consider each node of $C_i$ sequentially.
        In what follows, let $j$ with $0 < j < t_i-1$.
        Let $j_L = \left\lfloor \frac{j}{2}\right\rfloor + 1$, and let $j_R = \left\lceil \frac{t_1-1+j}{2}\right\rceil - 1$.
        We compute:
        \begin{align*}
            \Lambda(v) = \max\{&\Lambda_i, \lambda(i,j_R) - \lambda(i,j_L-1), \\
            &\ell_\pi(w_0^i,S) - \lambda(i,j_{\lim}^i) + \lambda(i,j_L-1), \\
            &\ell_\pi(w_{t_i-1}^i,S) - \lambda(i,j_R) + \lambda(i,j_{\lim}^i) \}
        \end{align*}
    \end{enumerate}
    \item Finally, we output a node $v \notin S$ such that $\Lambda(v)$ is minimized.
\end{enumerate}
{\bf Complexity.}
By Proposition~\ref{prop:compute-voronoi}, Step $1$ of the algorithm can be done in ${\cal O}(n)$ time.
Step $2$ can be also done in ${\cal O}(n)$ time.
Furthermore, since we have $|S_i| \leq 2$ for every $i$ with $1 \leq i \leq q$, we can compute all values $\Lambda_1,\Lambda_2,\ldots,\Lambda_q$ in ${\cal O}(n)$ time, simply by keeping track of the $\max\{3,p\}$ nodes of $S$ with maximum load. 
Let $i$ with $1 \leq i \leq q$ be fixed.
By dynamic programming, we compute in ${\cal O}(t_i)$ time the values $\lambda(i,j)$, for every $j$ with $0 \leq j \leq t_i-1$.
Then, we can compute in ${\cal O}(t_i)$ time the values $\Lambda(v)$, for every $v \in C_i$, simply by scanning once the path $P_i$ and performing ${\cal O}(1)$ operations for each node.
Since we have $\sum_{i=1}^q t_i = {\cal O}(n)$, the total running time for Step $3$ is also in ${\cal O}(n)$.

\smallskip
\noindent
{\bf Correctness.}
Let $i$ be such that $1 \leq i \leq q$, and let $v \in C_i$ be arbitrary.
We prove in what follows that $\Lambda(v) = L_\pi(S+v)$.
First we claim that $C_i \cup S_i \subseteq \bigcup\{\mathbb{T}(s,S) \mid s \in S_i\}$.
Indeed, suppose by contradiction that some node $u \in C_i$ belongs to $\mathbb{T}(s',S)$, for some $s' \in S \setminus S_i$.
The unique $us'$-path intersects $S_i$, which by Lemma~\ref{lem:metric-interval} implies that $\mathbb{T}(s',S) \cap S_i \neq \emptyset$, a contradiction.
In the same way, we have $\mathbb{T}(v,S+v) \subseteq C_i \subseteq \bigcup\{\mathbb{T}(x,S+v) \mid x \in S_i \cup \{v\}\}$.
This implies that for every $s \in S_i$ we have $\mathbb{T}(s,S+v) \setminus C_i = \mathbb{T}(s,S) \setminus C_i$, while for every $s' \in S \setminus S_i$, we have $\mathbb{T}(s',S) = \mathbb{T}(s',S+v)$.
Hence, we obtain: 
\begin{align*}
    L_\pi(S+v) &= \max\{\ell_\pi(x,S+v) \mid x \in S_i \cup \{v\}\} \cup \{ \ell_\pi(s',S+v) \mid s' \in S \setminus S_i \} \\
    &= \max\{\Lambda_i\} \cup \{\ell_\pi(x,S+v) \mid x \in S_i \cup \{v\}\}
\end{align*}
Clearly, $\mathbb{T}(v,S+v) = \{u \in C_i \mid d(u,v) < d(u,S_i)\}$.
Let $j$ with $0 \leq j \leq t_1-1$ and $v=w_j^i$.
There are two cases to be considered:
\begin{itemize}
    \item {\it Case $|S_i| = 1$}. By symmetry, we may further assume $i=1, \ w_0^1 \notin S_1$. Then, $C_1 \subseteq \mathbb{T}(w_{t_1-1}^1,S)$, and by Lemma~\ref{lem:struct-territory} we have $\mathbb{T}(v,S+v) = W(v,w_{t_1-1}^1) = \{w_0^1,w_1^1,\ldots,w_{j_R}^1\}$. In particular, $\ell_\pi(v,S+v) = \sum_{k=0}^{j_R}\pi(w_k^1) = \lambda(1,j_R)$, and so $\ell_\pi(w_{t_1-1}^1,S+v) = \ell_\pi(w_{t_1-1}^1,S) - \lambda(1,j_R)$.
    \item {\it Case $|S_i| = 2$}. In this situation, $\mathbb{T}(v,S+v) = \{w_{j_L}^i,w_{j_L+1}^i,\ldots,w_{j_R}^i\}$. Therefore, $\ell_\pi(v,S+v) = \sum_{k=j_L}^{j_R}\pi(w_k^i) = \lambda(i,j_R) - \sum_{k=0}^{j_L-1}\pi(w_k^i) = \lambda(i,j_R) - \lambda(i,j_L-1)$. By the definition of $j_{\lim}^i$ we have $\mathbb{T}(w_0^i,S) \cap C_i = \{w_1^i,\ldots,w_{j_{\lim}^i}^i\}$ while $\mathbb{T}(w_{t_i-1}^i,S) \cap C_i = \{w_{j_{\lim}^i+1}^i,\ldots,w_{t_i-2}^i\}$. 
    
    Furthermore, we claim that $j_L-1 \leq j_{\lim}^i \leq j_R$. 
    Indeed, let us first assume that $j \leq j_{\lim}^i$. In particular, $j_L-1 \leq j \leq j_{\lim}^i$. For every $k$ with $j \leq k \leq j_{\lim}^i$, we have $d(v,w_k^i) < d(w_{t_i-1}^i,w_k^i)$, which follows from $d(w_0^i,w_k^i) \leq d(w_{t_i-1}^i,w_k^i)$ and $v$ is on the $w_0^iw_k^i$-path. Therefore, $j_{\lim}^i \leq j_R$. 
    From now on, we assume that $j_{\lim}^i < j$. In particular, $j_{\lim}^i < j_R$. For every $k$ with $j_{\lim}^i < k \leq j$, we have $d(v,w_k^i) < d(w_{0}^i,w_k^i)$, which follows from $d(w_0^i,w_k^i) \geq d(w_{t_i-1}^i,w_k^i)$ and $v$ is on the $w_{t_i-1}^iw_k^i$-path. Therefore, $j_L \leq j_{\lim}^i +1$.
    
    As a result,
    \begin{align*}
        \ell_\pi(w_0^i,S+v) &= \ell_\pi(w_0^i,S) - \sum_{k=j_L}^{j_{\lim}^i} \pi(w_k^i) \\
        &= \ell_\pi(w_0^i,S) - \left(\lambda(i,j_{\lim}^i) - \lambda(i,j_L-1)\right)
    \end{align*}
    and
    \begin{align*}
        \ell_\pi(w_{t_i-1}^i,S+v) &= \ell_\pi(w_{t_i-1}^i,S) - \sum_{k=j_{\lim}^i+1}^{j_R} \pi(w_k^i) \\
        &= \ell_\pi(w_{t_i-1}^i,S) - \left(\lambda(i,j_R) - \lambda(i,j_{\lim}^i) \right)
    \end{align*}
\end{itemize}
Overall, we proved that in both cases $\Lambda(v) = L_\pi(S+v)$.
Therefore, the algorithm is correct.
\end{proof}

Up to slight adjustments ({\it e.g.}, in the special case where $|S|=1$), we can directly apply to cycles our techniques for paths.
We prove it below.

\begin{lemma}\label{lem:cycle}
We can solve {\sc Balanced Vertex} in ${\cal O}(n)$ time if $G$ is a cycle.
\end{lemma}
\begin{proof}
Let us first assume that $p=|S|=1$.
Let $(v_0,v_1,\ldots,v_{n-1})$ be a cyclic ordering of the vertices, and let us assume without loss of generality that $S=(v_0)$.
We compute by dynamic programming the values $\lambda(j) = \sum_{k=0}^j \pi(v_k)$, for every $j$ with $1 \leq j < n$. It takes ${\cal O}(n)$ time.
By symmetry, we only consider in what follows the vertices $v_i$ with $0 < i \leq n/2$.
In this situation, by Lemma~\ref{lem:struct-territory}, 
$$\mathbb{T}(v_i,S+v_i) = W(v_i,v_0) = \{v_{\left\lfloor \frac i 2 \right\rfloor + 1},v_{\left\lfloor \frac i 2 \right\rfloor + 2}, \ldots,v_{\left\lceil\frac{n-i}{2}\right\rceil-1}\}.$$
Let $\ell_i = \lambda(\left\lceil\frac{n-i}{2}\right\rceil-1) - \lambda(\left\lfloor \frac i 2 \right\rfloor)$.
Note that $L_{\pi}(S+v_i) = \max\{\ell_i,\pi(G)-\ell_i\}$.
Therefore, it suffices to output a vertex $v_i$ with $1 \leq i \leq n/2$ and such that $ \max\{\ell_i,\pi(G)-\ell_i\}$ is minimized.
The running time is in ${\cal O}(n)$.

From now on, we assume that $p=|S| >1$.
We compute $\mathbf{Vor}(G,S)$ and the loads $\ell_\pi(s_j,S)$, for every $j$ with $1 \leq j \leq p$.
By Proposition~\ref{prop:compute-voronoi}, this can be done in ${\cal O}(n)$ time.
Let $C_1,C_2,\ldots,C_q$ be the connected components of $G \setminus S$.
For every $i$ with $1 \leq i \leq q$, let $S_i = N(C_i)$, and let $P_i$ be the subpath induced by $N[C_i] = C_i \cup S_i$.
Note that the two vertices of $S_i$ are the two ends of $P_i$.
We are done applying the exact same procedure for each subpath $P_i$ as in Lemma~\ref{lem:path} (Step $3$ of the algorithm for paths).
The running time is in ${\cal O}(n)$.
Correctness follows from the exact same arguments as in the proof of Lemma~\ref{lem:path}, and it is even simpler because the case $|S_i|=1$ cannot occur.
\end{proof}

\subsection{Trees}\label{sec:tree}

We now address the case of general trees.
Unlike for paths, the number of former sites in $S$ whose territory is modified after insertion of a new site may be arbitrarily large.
See Fig.~\ref{fig:tree} for an illustration.
Fortunately, we can use {\em centroid decomposition} in order to compute, for every $v \notin S$, a set of ${\cal O}(\log{n})$ nodes that are on the paths between $v$ and the territories of all but ${\cal O}(\log{n})$ former sites in $S$.
More precisely, being given a tree $T$, we can summarize our strategy as follows:
\begin{enumerate}
    \item We compute a centroid, {\it i.e.} a node $c$ such that every connected component of $T \setminus \{c\}$ contains at most $n/2$ nodes. 
    \item We process the unique site $s_i \in S$ such that $c \in \mathbb{T}(s_i,S)$ (see Lemma~\ref{lem:tree-1});
    \item We process all sites $s_j$ whose territories are on different connected components of $T \setminus \{c\}$ than $v$ (see Lemmas~\ref{lem:tree-2} and~\ref{lem:tree-3});
    \item Finally, we recursively apply our strategy in order to process all remaining sites whose territories are on the same connected component of $T \setminus \{c\}$ as $v$.
\end{enumerate}

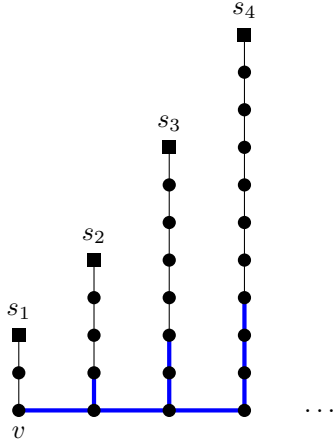
\begin{figure}[!ht]
    \centering
    \begin{tikzpicture}
        
       \draw (0,0) -- (0,1);
       \draw (1,0) -- (1,2);
       \draw (2,0) -- (2,3.5);
       \draw (3,0) -- (3,5);
       
       \draw[ultra thick, blue] (0,0) -- (3,0);
       \draw[ultra thick, blue] (1,0) -- (1,.5);
       \draw[ultra thick, blue] (2,0) -- (2,1);
       \draw[ultra thick, blue] (3,0) -- (3,1.5);
        
       \node[circle,fill=black,inner sep=0pt,minimum size=5pt,label=below:{$v$}] at (0,0) {};
       \node[circle,fill=black,inner sep=0pt,minimum size=5pt,label=below:{}] at (1,0) {};
       \node[circle,fill=black,inner sep=0pt,minimum size=5pt,label=below:{}] at (2,0) {};
       \node[circle,fill=black,inner sep=0pt,minimum size=5pt,label=below:{}] at (3,0) {};
       \node at (4,0) {$\ldots$};
       
       \node[rectangle,fill=black,inner sep=0pt,minimum size=5pt,label=above:{$s_1$}] at (0,1) {};
       \node[rectangle,fill=black,inner sep=0pt,minimum size=5pt,label=above:{$s_2$}] at (1,2) {};
       \node[rectangle,fill=black,inner sep=0pt,minimum size=5pt,label=above:{$s_3$}] at (2,3.5) {};
       \node[rectangle,fill=black,inner sep=0pt,minimum size=5pt,label=above:{$s_4$}] at (3,5) {};
    
       \node[circle,fill=black,inner sep=0pt,minimum size=5pt,label=below:{}] at (0,.5) {};
        
       \node[circle,fill=black,inner sep=0pt,minimum size=5pt,label=below:{}] at (1,.5) {};
       \node[circle,fill=black,inner sep=0pt,minimum size=5pt,label=below:{}] at (1,1) {};
       \node[circle,fill=black,inner sep=0pt,minimum size=5pt,label=below:{}] at (1,1.5) {};;
       
       \node[circle,fill=black,inner sep=0pt,minimum size=5pt,label=below:{}] at (2,.5) {};
       \node[circle,fill=black,inner sep=0pt,minimum size=5pt,label=below:{}] at (2,1) {};
       \node[circle,fill=black,inner sep=0pt,minimum size=5pt,label=below:{}] at (2,1.5) {};
       \node[circle,fill=black,inner sep=0pt,minimum size=5pt,label=below:{}] at (2,2) {};
       \node[circle,fill=black,inner sep=0pt,minimum size=5pt,label=below:{}] at (2,2.5) {};
       \node[circle,fill=black,inner sep=0pt,minimum size=5pt,label=below:{}] at (2,3) {};
       
       \node[circle,fill=black,inner sep=0pt,minimum size=5pt,label=below:{}] at (3,.5) {};
       \node[circle,fill=black,inner sep=0pt,minimum size=5pt,label=below:{}] at (3,1) {};
       \node[circle,fill=black,inner sep=0pt,minimum size=5pt,label=below:{}] at (3,1.5) {};
       \node[circle,fill=black,inner sep=0pt,minimum size=5pt,label=below:{}] at (3,2) {};
       \node[circle,fill=black,inner sep=0pt,minimum size=5pt,label=below:{}] at (3,2.5) {};
       \node[circle,fill=black,inner sep=0pt,minimum size=5pt,label=below:{}] at (3,3) {};
       \node[circle,fill=black,inner sep=0pt,minimum size=5pt,label=below:{}] at (3,3.5) {};
       \node[circle,fill=black,inner sep=0pt,minimum size=5pt,label=below:{}] at (3,4) {};
       \node[circle,fill=black,inner sep=0pt,minimum size=5pt,label=below:{}] at (3,4.5) {};
       
    \end{tikzpicture}
    \caption{Beginning of the construction of a tree with $p$ former sites in $S$ and ${\cal O}(p^2)$ nodes. There is a central path with $p$ nodes whose first node is $v$. Each former site $s_i$, with $2 \leq i \leq p$, is connected to the central path by a path of length $3i+1$. The thick edges represent the paths between the new site $v$ and every node in its territory. In particular, for every $i$ with $2 \leq i \leq p$, exactly one node is removed from the territory of $s_i$.}
    \label{fig:tree}
\end{figure}

\begin{lemma}\label{lem:tree-1}
Let $s$ be an arbitrary node in a tree $T=(V,E)$.
Let $\pi : V \to \mathbb{R}_{\geq 0}$ be any cost function.
We can compute $\alpha(v) = \pi\left(W(v,s)\right)$, for every node $v \in V$, in ${\cal O}(n)$ time.
\end{lemma}
\begin{proof}
We root $T$ at $s$. Then, we preprocess $T$ in ${\cal O}(n)$ time so that the following holds for every node $v$:
\begin{itemize}
    \item We store $d_T(v,s)$ --- the level of node $v$;
    \item We store $\pi(T_v)$ --- the cost of all nodes in the subtree $T_v$ rooted at $v$;
    \item For every $i$ with $0 \leq i \leq d_T(v,s)$, we can compute the $i^{th}$-level ancestor of $v$ in ${\cal O}(1)$ time~\cite{BeF04}.
\end{itemize}
We consider each node $v \in V$ sequentially.
Let $d_T(v,s) = 2q+\varepsilon$, with $\varepsilon \in \{0,1\}$.
In this situation, $W(v,s) = T_x$, with $x$ the $(q+1)^{th}$-level ancestor of $v$.
In particular, we can compute $x$, and so, $\pi(T_x)=\alpha(v)$, in ${\cal O}(1)$ time.
Overall, the total running time is in ${\cal O}(n)$.
\end{proof}

Let $c$ be an arbitrary cut-vertex in a tree $T$ (in our Theorem~\ref{thm:tree}, this node $c$ will always be a centroid of the tree).
The next two lemmas allow us to process, for all potential locations $v$ for a new site, all former sites in $S$ whose territory is disconnected from $v$ in $T \setminus \{c\}$.
In particular, using Lemma~\ref{lem:tree-2} we can compute the cost of all nodes in those territories that are closer to $v$ than to $S$ (and so, that would end up in $v$'s territory).
Using Lemma~\ref{lem:tree-3}, we can compute the maximum load of those sites in the Voronoi diagram of $T$ with respect to $S+v$.

\begin{lemma}\label{lem:tree-2}
Let $c$ be an internal node in a tree $T=(V,E)$, and let $S$ be a proper vertex list.
We can compute $\beta(v) = \sum\{\pi(x) \mid c \in I(v,x) \ \text{and} \ d(x,v) < d(x,S)\}$, for every node $v \in V \setminus \{c\}$, in ${\cal O}(n)$ time.
\end{lemma}
\begin{proof}
Let us first describe our algorithm:
\begin{enumerate}
    \item In ${\cal O}(n)$ time, we compute $d(x,S)$ and $d(x,c)$ for every node $x$.
    \item For every $i$ with $0 \leq i \leq n$, we compute $a[i] = \sum\{ \pi(x) \mid d(x,c) + i < d(x,S) \}$.
    For that, we start computing $b[j] = \sum\{ \pi(x) \mid d(x,S) = d(x,c) + j \}$, for every $j$ with $0 \leq j \leq n$, simply by scanning once all nodes $x$.
    Since $a[i] = \sum_{j > i}b[j]$, we are done in additional ${\cal O}(n)$ time by dynamic programming.
    \item We consider each connected component $C_1,C_2,\ldots,C_q$ of $T \setminus \{c\}$ sequentially. Let $k$ with $1 \leq k \leq q$ be fixed.
    For every $i$ with $0 \leq i \leq |C_k|$, we compute $a_k[i] = \sum\{ \pi(x) \mid x \in C_k, \ d(x,c) + i < d(x,S)  \}$.
    This can be done in ${\cal O}(|C_k|)$ time by using the exact same approach as for the prior step of the algorithm.
    Finally, for every node $v \in C_k$, we set: $$\beta(v) =  a[d(v,c)] - a_k[d(v,c)].$$
    This can be computed in ${\cal O}(1)$ time per node of $C_k$, and therefore, in ${\cal O}(|C_k|)$ time.
\end{enumerate}
The total running time of the algorithm is in ${\cal O}(n)$.
In order to prove correctness of the algorithm, let $v \in V \setminus \{c\}$ be arbitrary.
We have $v \in C_k$ for some $k$ with $1 \leq k \leq q$.
For every node $x$, we have $c \in I(x,v)$ if and only if $x \notin C_k$.
Furthermore, in this situation we have $d(x,v) < d(x,S)$ if and only $d(x,c) + d(c,v) < d(x,S)$. 
The sum of the costs of all nodes $x$ such that $d(x,c) + d(c,v) < d(x,S)$ is exactly $a[d(v,c)]$.
In order to compute $\beta(v)$ from $a[d(v,c)]$, we need to remove from the latter the sum of the costs of all such nodes $x$ with $x \in C_k$, that is exactly $a_k[d(v,c)]$.
\end{proof}

\begin{lemma}\label{lem:tree-3}
Let $c$ be an internal node in a tree $T=(V,E)$, and let $S$ be a proper vertex list.
We can compute $\gamma(v) = \max\{ \ell_\pi(s,S+v) \mid c \notin \mathbb{T}(s,S), \ c \in I(v,s)  \}$, for every node $v \in V \setminus \{c\}$, in ${\cal O}(n)$ time.
\end{lemma}
\begin{proof}Let $C_1,C_2,\ldots,C_q$ be the connected components of $T \setminus \{c\}$.
For convenience, given any node $v$ with $v \neq c$, we denote by $C(v)$ the unique component $C_i$ that contains $v$, for some $i$ with $1 \leq i \leq q$.
We denote by $S'$ the sublist obtained from $S$ by erasing the unique site $s_c$ such that $c \in \mathbb{T}(s_c,S)$.
For every $j$ with $0 \leq j \leq n$, let us define, for every node $s \in S'$, $$\ell_j(s) = \ell_\pi(s,S) - \sum \{ \pi(x) \mid x \in \mathbb{T}(s,S) \ \text{and} \ d(x,s) > d(x,c) + j \}.$$
Intuitively, $\ell_j(s)$ would be the load of $s$ upon insertion of a new node $v$, with $v \notin C(s)$ and $d(v,c) = j$.
In particular, the following holds for every node $v$, with $v \neq c$ and $d(v,c) = j$: $\gamma(v) = \max\{ \ell_j(s) \mid s \in S' \setminus C(v) \}$.
Therefore, in order to compute $\gamma(v)$, for every node $v$ with $v \neq c$, it suffices to compute the following information for every $j$ with $0 \leq j \leq n$:
\begin{itemize}
    \item a node $s_j$ maximizing $\ell_j(s_j)$ within $S'$;
    \item a node $s_j'$ maximizing $\ell_j(s_j')$ within $S' \setminus C(s_j)$.
\end{itemize}
For that, we first compute $\mathbf{Vor}(T,S)$, and the distances $d(x,S)$ and $d(x,c)$ for every node $x$.
By Proposition~\ref{prop:compute-voronoi}, this can be done in ${\cal O}(n)$ time.
We order the nodes $x$, with $x \neq c$, by nondecreasing values $d(x,S) - d(x,c)$, which can also be done in ${\cal O}(n)$ time by using counting sort.
For every node $s \in S'$, let $\lambda(s) = \ell_{\pi}(s,S)$. Intuitively, at the end of each step $j$ of our main procedure (presented next), we must have $\lambda(s) = \ell_j(s)$.
We now proceed as follows for every $j$ with $0 \leq j \leq n$:
\begin{itemize}
    \item {\it Case $j = 0$ (first step).} Let $X_0 = \{ x \in V \setminus \{c\} \mid d(x,S) > d(x,c) \}$. We consider each node $x \in X_0$ sequentially. Let $s \in S$ be such that $x \in \mathbb{T}(s,S)$. If $c \notin \mathbb{T}(s,S)$, then we set $\lambda(s) = \lambda(s) - \pi(x)$.
    Finally, we compute a vertex $s_0$ ($s_0'$, resp.) that maximizes $\lambda(s_0)$ ($\lambda(s_0')$, resp.) within $S'$ (within $S' \setminus C(s_0)$, resp.).
    \item {\it Case $j > 0$.} Let $X_j = \{ x \in V \setminus \{c\} \mid d(x,S) - d(x,c) = j\}$. We scan each node in $X_j$, computing along the way a subset $S_j$ of relevant nodes in $S'$. Initially, we set $S_j = \emptyset$.  We consider each node $x \in X_j$ sequentially. Let $s \in S$ be such that $x \in \mathbb{T}(s,S)$. If $c \notin \mathbb{T}(s,S)$, then we add $s$ in $S_j$, then we set $\lambda(s) = \lambda(s) + \pi(x)$. Finally, we compute a vertex $s_{j}$ ($s_j'$, resp.) that maximizes $\lambda(s_{j})$ ($\lambda(s_j')$, resp.) within $S_j \cup \{s_{j-1}\}$ (within $\left(S_j \cup \{s_{j-1},s_{j-1}'\}\right) \setminus C(s_{j})$, resp.).
\end{itemize}
Step $0$ can be done in ${\cal O}(n)$ time.
Furthermore, since we ordered all nodes $x$ with $x \neq c$ by nondecreasing values $d(x,S)-d(x,c)$, all subsets $X_1,X_2,\ldots,X_n$ can be computed in ${\cal O}(n)$ time.
Being given $X_j$, for some $j$ with $1 \leq j \leq n$, Step $j$ can be done in ${\cal O}(|X_j|+|S_j|) = {\cal O}(|X_j|)$ time, provided we can compute in ${\cal O}(1)$ time, for every $x \in X_j$ the unique $s \in S$ such that $x \in \mathbb{T}(s,S)$.
That is indeed the case because we precomputed $\mathbf{Vor}(T,S)$.
Overall, the running time is in ${\cal O}(\sum_j |X_j|) = {\cal O}(|X_0|) = {\cal O}(n)$.

\smallskip
In order to prove correctness of the algorithm, we prove (as claimed above) that at the end of each step $j$, we have $\lambda(s) = \ell_j(s)$, for every $s \in S'$. 
We prove it by induction on $j$.
If $j=0$, then at the end of step $0$ we have:
\begin{align*}
    \lambda(s) &= \ell_\pi(s) - \sum\{\pi(x) \mid x \in X_0 \cap \mathbb{T}(s,S)\} \\
    &= \ell_\pi(s) - \sum\{\pi(x) \mid x \in \mathbb{T}(s,S), \ d(x,S) > d(x,c)\} \\
    &= \ell_\pi(s) - \sum\{\pi(x) \mid x \in \mathbb{T}(s,S), \ d(x,s) > d(x,c)\} \\
    &= \ell_0(s)
\end{align*}
From now on, we assume that $j > 0$. By the induction hypothesis, we have $\lambda(s) = \ell_{j-1}(s)$ before we start step $j$.
Therefore, at the end of step $j$ we have:
\begin{align*}
    \lambda(s) &= \ell_{j-1}(s) + \sum\{\pi(x) \mid x \in X_j \cap \mathbb{T}(s,S)\} \\
    &= \ell_{j-1}(s) + \sum\{\pi(x) \mid x \in \mathbb{T}(s,S), \ d(x,S) = d(x,c) + j\} \\
    &= \ell_\pi(s) - \sum\{\pi(x) \mid x \in \mathbb{T}(s,S), \ d(x,s) > d(x,c) + j-1\} + \sum\{\pi(x) \mid x \in \mathbb{T}(s,S), \ d(x,S) = d(x,c) + j\} \\
    &= \ell_\pi(s) - \sum\{\pi(x) \mid x \in \mathbb{T}(s,S), \ d(x,s) > d(x,c) + j\} \\
    &= \ell_j(s)
\end{align*}
Finally, we need to prove the maximality of the selected nodes $s_{j},s_{j}'$ at each step $j$.
We also prove it by induction on $j$.
This is straightforward if $j=0$ because we consider all nodes of $S'$ (of $S' \setminus C(s_0)$, resp.) in order to compute $s_0$ ($s_0'$, resp.).
From now on, we assume that $j > 0$.
Suppose by contradiction the existence of some node $s \in S'$, with $s \neq s_{j}$ and $\lambda_j(s) > \lambda_j(s_{j})$.
In particular, $s \notin S_j$, which implies $\lambda_j(s) = \lambda_{j-1}(s)$.
But then, by induction on $j$,
\begin{align*}
    \lambda_j(s) = \lambda_{j-1}(s) \leq \lambda_{j-1}(s_{j-1}) \leq \lambda_j(s_{j-1}) \leq \lambda_j(s_j).
\end{align*}
A contradiction.
We can prove similarly as above that $s_j'$ maximizes $\ell_j(s_j')$ within $S' \setminus C(s_j)$.
\end{proof}

We are now ready to prove the main result in this subsection:

\begin{theorem}\label{thm:tree}
We can solve {\sc Balanced Vertex} in ${\cal O}(n\log{n})$ time for trees.
\end{theorem}
\begin{proof}
Let $T=(V,E)$ be an $n$-node tree, let $\pi$ be any nonnegative cost function and let $S$ be a proper vertex list.
We consider a more general problem, where every node $v \notin S$ is further assigned nonnegative values $\lambda(v),\Lambda(v)$ (initially, both values equal $0$).
Intuitively, we account for sites and territories in a supertree of $T$ but {\em not} in the tree $T$ that we are currently considering.
Let us define: $$R(v) = \max\{\ell_\pi(v,S+v)+\lambda(v),\Lambda(v)\} \cup \{\ell_\pi(s,S+v) \mid s \in S\}.$$
Our algorithm in what follows compute {\em all} values  $R(v)$, for every $v \notin S$.
We may assume $S \neq \emptyset$ (otherwise, $R(v) = \max\{\lambda(v)+\pi(T),\Lambda(v)\}$ for every node $v$).
If $V\setminus S$ is reduced to a unique node $v$, then we first compute $\mathbf{Vor}(T,S+v)$, which by Proposition~\ref{prop:compute-voronoi} can be done in ${\cal O}(n)$ time.
Doing so, we can compute in ${\cal O}(n)$ time the values $\ell_\pi(v,S+v)$ and $\ell_\pi(s,S+v)$, for every $s \in S$, and so, $R(v)$.
Thus, from now on we assume that $|V \setminus S| > 1$.
Our algorithm goes as follows:
\begin{enumerate}
    \item We compute $\mathbf{Vor}(T,S)$. By Proposition~\ref{prop:compute-voronoi}, this can be done in ${\cal O}(n)$ time.
    \item We compute a {\em centroid}, {\it i.e.} a node $c$ such that every connected component of $T \setminus \{c\}$ contains at most $n/2$ nodes. It can be done in ${\cal O}(n)$ time~\cite{Gol71}. 
    If furthermore $c \notin S$, then we compute $R(c)$. It can be done in ${\cal O}(n)$ time if we are given $\mathbf{Vor}(T,S+c)$. Therefore, by Proposition~\ref{prop:compute-voronoi}, we can compute $R(c)$ in ${\cal O}(n)$ time. 
    \item Let $s_c \in S$ be such that $c \in \mathbb{T}(s_c,S)$. We compute $\alpha(v) = \sum\{ \pi(x) \mid x \in W(v,s_c) \cap \mathbb{T}(s_c,S)\}$, for every $v \notin S$, which can be done in ${\cal O}(n)$ time as follows: we set the cost of every node in $V \setminus \mathbb{T}(s_c,S)$ to $0$, then we apply Lemma~\ref{lem:tree-1}. Afterwards, we replace $\pi$ with some new cost function $\pi'$ so that: $$\pi'(x) = \begin{cases} \pi(x) \ \text{if} \ x \notin \mathbb{T}(s_c,S) \\ 0 \ \text{otherwise.}\end{cases}$$
    \item For every $v$ with $v \notin S \cup \{c\}$, we compute the cost of all vertices $x \in \mathbb{T}(v,S+v)$ such that $v,x$ are in separate connected components of $T \setminus \{c\}$. Equivalently, we aim at computing $\beta(v) = \sum\{ \pi'(x) \mid c \in I(v,x) \ \text{and} \ d(x,v) < d(x,S) \}$, which can be done in ${\cal O}(n)$ time if we apply Lemma~\ref{lem:tree-2}.
    For every $v$ with $v \notin S \cup \{c\}$, we further compute the maximum cost of a subset $\mathbb{T}(s,S+v)$, amongst all sites $s$ with $s \neq s_c$ and $v,\mathbb{T}(s,S)$ are in separate connected components of $T \setminus \{c\}$. Equivalently, we aim at computing $\gamma(v) = \max\{ \ell_\pi(s,S+v) \mid c \notin \mathbb{T}(s,S), \ c \in I(v,s)  \}$, which can be done in ${\cal O}(n)$ time if we apply Lemma~\ref{lem:tree-3}.
    \item Let $C_1,C_2,\ldots,C_q$ be the connected components of $T \setminus \{c\}$. We consider each component $C_j$ sequentially, for every $j$ with $1 \leq j \leq q$.
    Let $S_j = (S \setminus \{s_c\}) \cap C_j$. Let $\pi_j$ be the restriction of $\pi'$ to $C_j$.
    For every $v \in C_j$, we define:
    \begin{align*}
        \lambda_j(v) &= \lambda(v) + \alpha(v) + \beta(v)\\
        \Lambda_j(v) &= \max\{\Lambda(v), \ell_\pi(s_c,S) - \alpha(v), \gamma(v)\} \\
        R_j(v) &= \max\{\lambda_j(v)+\ell_{\pi_j}(v,S_j+v),\Lambda_j(v)\} \cup \{ \ell_{\pi_j}(s,S_j+v) \mid s \in S_j \}
    \end{align*}
    We compute $R_j(v)$ for every $v \in C_j$. For that, it suffices to apply our algorithm recursively on $C_j,\pi_j,S_j,\lambda_j,\Lambda_j$.
    Finally, we set $R(v) = R_j(v)$ for every $v \in C_j$.
\end{enumerate}
The recursion depth is in ${\cal O}(\log{n})$ because at each stage the maximum number of nodes in a subtree considered is halved.
Since all subtrees considered at any recursion stage are node disjoint, the running time of the stage is in ${\cal O}(n)$.
Hence, the total running time is in ${\cal O}(n\log{n})$.

\smallskip
In order to prove correctness of the algorithm, it suffices to prove that for every $j$ with $1 \leq j \leq q$, for every $v \in C_j$, $R_j(v) = R(v)$.
For that, by Lemma~\ref{lem:struct-territory}:
\begin{align*}
    \mathbb{T}(v,S+v) &= \bigcup\{ x \in V \mid d(x,v) < d(x,S) \} \\
    &= \bigcup\{ W(v,s) \cap \mathbb{T}(s,S) \mid s \in S \}
\end{align*}
In particular,
\begin{align*}
    \lambda(v) + \ell_\pi(v,S+v) &= \lambda(v) + \sum_{s \in S}\pi\left(W(v,s) \cap \mathbb{T}(s,S)\right) \\
    &= \lambda(v) + \pi\left(W(v,s_c) \cap \mathbb{T}(s_c,S)\right) + \sum_{s \in S \setminus \{s_c\}}\pi\left(W(v,s) \cap \mathbb{T}(s,S)\right) \\
    &= \lambda(v) + \alpha(v) + \sum_{s \in S \setminus \{s_c\}}\pi'\left(W(v,s) \cap \mathbb{T}(s,S)\right) \\
    &= \lambda(v) + \alpha(v) + \sum_{s \in S \setminus (S_j \cup \{s_c\})}\pi'\left(W(v,s) \cap \mathbb{T}(s,S)\right) + \sum_{s \in S_j}\pi'\left(W(v,s) \cap \mathbb{T}(s,S)\right) \\
    &= \lambda(v) + \alpha(v) + \beta(v) + \sum_{s \in S_j}\pi_j\left(W(v,s) \cap \mathbb{T}(s,S_j)\right) \\
    &= \lambda_j(v) + \ell_{\pi_j}(v,S_j+v)
\end{align*}
For every $s \in S$, we have that $\mathbb{T}(s,S+v) = \mathbb{T}(s,S) \setminus W(v,s)$.
Therefore,
\begin{align*}
    \max\{\Lambda(v)\}\cup\{\ell_\pi(s,S+v) \mid s \in S\} &= \max\{\Lambda(v),\ell_\pi(s_c,S+v)\} \cup \{\ell_\pi(s,S+v) \mid s \in S\setminus\{s_c\}\} \\
    &= \max\{\Lambda(v),\ell_\pi(s_c,S) - \alpha(v)\} \cup \{\ell_\pi(s,S+v) \mid s \in S\setminus (S_j \cup \{s_c\})\} \\ &\cup \{\ell_\pi(s,S+v) \mid s \in S_j \} \\
    &= \max\{\Lambda(v),\ell_\pi(s_c,S) - \alpha(v),\gamma(v)\} \cup \{\ell_{\pi}(s,S+v) \mid s \in S_j \} \\
    &= \max\{\Lambda_j(v)\} \cup  \{\ell_{\pi_j}(s,S_j+v) \mid s \in S_j \}
\end{align*}
By combining both equalities, we obtain as desired $R(v) = R_j(v)$.
\end{proof}

\begin{remark}\label{rk:tree}
In the recursive algorithm presented for trees, a site $s_c$ is discarded from $S$ at every recursion stage. Therefore, the running time of the algorithm is in ${\cal O}(n|S|)$, that is linear if $|S|$ is a constant. This can be improved to ${\cal O}(n\log{|S|})$ as follows: at every recursion stage, we compute a node $c$ so that there are at most $|S|/2$ sites in every connected component of $T \setminus \{c\}$. It can be done in ${\cal O}(n)$ time because it is a special case of {\em weighted centroid} computation~\cite{Gol71}. In doing so, the recursion depth goes down to ${\cal O}(\log{|S|})$.
\end{remark}

\subsection{Bounded-treewidth graphs}\label{sec:tw}

We generalize the results of the prior subsections to some increasing hierarchy of graph classes, namely, graphs of bounded {\em treewidth}.
Recall that a tree decomposition of a graph $G$ is a pair $(T,{\cal X})$, where $T$ is a tree and ${\cal X} : V(T) \mapsto 2^{V(G)}$ is a mapping that satisfies the following properties:
\begin{itemize}
    \item For every vertex $v \in V(G)$, there exists a node $t \in V(T)$ such that $v \in {\cal X}_t$;
    \item For every edge $uv \in E(G)$, there exists a node $t \in V(T)$ such that $u,v \in {\cal X}_t$;
    \item For every vertex $v \in V(G)$, the node subset $\{ t \in V(T) \mid v \in {\cal X}_t \}$ induces a connected subtree of $T$.
\end{itemize}
The subsets ${\cal X}_t$, for every node $t \in V(T)$, are called the bags of the tree decomposition.
We define the width of a tree decomposition as the largest size of its bags minus one.
The treewidth of $G$ is the minimum width of its tree decompositions.
In particular, graphs of unit treewidth are exactly the forests, and every cycle has treewidth two.

\medskip
Cabello and Knauer were the first to propose a reduction from some polynomial-time solvable distance problems on bounded-treewidth graphs to {\em orthogonal range queries}~\cite{CaK09}.
We propose a novel application of their framework to the {\sc Balanced Vertex} problem with constant number of sites in $S$.

We only consider counting queries in what follows.
In particular, a range query will always refer to a counting orthogonal range query.
Let ${\cal P}$ denote a static $n$-set of $k$-dimensional points in $\mathbb{R}^k$.
Let also $f : {\cal P} \mapsto \mathbb{R}$ be fixed.
A box is the Cartesian product of $k$ intervals.
Finally, a range query asks to compute $$\sum\{f(\overrightarrow{p}) \mid \overrightarrow{p} \in {\cal P} \cap {\cal B}\}$$ for some box ${\cal B}$.
Bentley presented the $k$-range tree data structure in order to answer to range queries~\cite{Ben79}.
The performance analysis of the latter has been recently improved:

\begin{lemma}[\cite{BHM20}]\label{lem:rq}
Let $B(n,k) = \binom{k+\log{n}}{k}$. 
Being given a fixed $n$-set of $k$-dimensional points, and any weight function $f$, a $k$-range tree can be constructed in ${\cal O}(k^2B(n,k)n)$ time.
Moreover, a range query can be answered in ${\cal O}(2^kB(n,k))$ time.
\end{lemma}

Throughout this subsection, we write all running times as some functions of $B(n,k)$. 
We refer to Remark~\ref{rk:tw} for a more detailed analysis of $B(n,k)$.

\medskip
Roughly, our strategy consists in mimicking our previous approach for trees.
Specifically, centroids in a tree decomposition of width at most $k$ are balanced separators of size at most $k$.
We generalize Lemma~\ref{lem:tree-2}, for cut-vertices in trees, to the following Lemma~\ref{lem:sep} for $k$-separators.
Our generalization makes use of range queries.
Intuitively, the function $D_S$ in the statement of Lemma~\ref{lem:sep} below represents the distances to the sites of $S$ in some supergraph of $G$.

\begin{lemma}\label{lem:sep}
Being given a graph $G=(V,E,\lambda)$ with positive integer edge weights, and two nonnegative cost functions $\pi$ and $D_S$, let $A$ and $B$ satisfy: $A \cup B = V$, $|A\cap B| \leq k$.
We can compute $\delta(a,B) = \sum\{\pi(b) \mid b \in B \setminus A \ \text{and} \ d(b,a) < D_S(b)\}$, for every $a \in A \setminus B$, in ${\cal O}(km+k2^kB(n,k)n)$ time.
\end{lemma}
\begin{proof}
For convenience, let $X = A \cap B$. We further assume for simplicity that $|X| = k$.
We first compute a shortest-path tree rooted at $x$, for every $x \in X$. Since all edge weights are positive integer, this can be done in ${\cal O}(m)$ time per vertex in $X$, using Thorup's single-source shortest-path algorithm~\cite{Tho99}. Let $X=(x_1,x_2,\ldots,x_k)$ be totally ordered in some arbitrary way. We proceed as follows for every $i$ with $1 \leq i \leq k$.
\begin{enumerate}
    \item For every $b \in B \setminus X$, we create a $k$-dimensional point $\overrightarrow{p_i(b)} = (p_{i,1}(b),p_{i,2}(b),\ldots,p_{i,k}(b))$, where:
    $$\begin{cases}
    p_{i,i}(b) = D_S(b) \\
    p_{i,j}(b) = d(b,x_j) - d(b,x_i) \ \text{for every} \ j \ \text{with} \ 1 \leq j \leq k, \ j \neq i
    \end{cases}$$
    Let $f(\overrightarrow{p_i(b)}) = \pi(b)$.
    We add all these points and their weights in some $k$-range tree. By Lemma~\ref{lem:rq}, this can be done in ${\cal O}(k^2B(n,k)n)$ time.
    \item For every $a \in A \setminus X$, we aim at computing the cost $\delta_i(a,B)$ of all vertices $b \in B \setminus X$ with the following properties: $D_S(b) > d(b,a)$, $x_i$ is on a shortest $ab$-path; no vertex of $\{x_1,x_2,\ldots,x_{i-1}\}$ is on a shortest $ab$-path. We can formulate these properties as range constraints for each coordinate of $\overrightarrow{p_i(b)}$, as follows: $$\begin{cases}
    p_{i,j}(b) = d(b,x_j) - d(b,x_i) > d(a,x_i) - d(a,x_j), \ \text{for every} \ j \ \text{with} \ 1 \leq j < i \\
    p_{i,i}(b) = D_S(b) > d(a,b) \\
    p_{i,j}(b) = d(b,x_j) - d(b,x_i) \geq d(a,x_i) - d(a,x_j), \ \text{for every} \ j \ \text{with} \ i < j \leq k
    \end{cases}$$
    As a result, for a fixed $a$, we can compute the cost of all corresponding vertices of $B \setminus X$ with a single range query, which by Lemma~\ref{lem:rq} can be answered in ${\cal O}(2^kB(n,k))$ time. 
\end{enumerate}
Finally, we compute $\delta(a,B) = \sum_{i=1}^k \delta_i(a,B)$, for every $a \in A \setminus X$.
\end{proof}

Being given a tree decomposition $(T,{\cal X})$ of width ${\cal O}(k)$, the following result is a combination of Lemma~\ref{lem:sep} with a centroid decomposition of $T$.

\begin{lemma}\label{lem:wall}
Let $S$ be an arbitrary vertex subset in a graph $G=(V,E)$, and let $\pi$ be any nonnegative cost function.
If $G$ has treewidth at most $k$, then we can compute $\delta(v) = \sum\{\pi(u) \mid d(u,v) < d(u,S)\}$, for every $v \in V \setminus S$, in ${\cal O}(k2^{{\cal O}(k)}B(n,k)n\log{n})$ time.
\end{lemma}
\begin{proof}
We consider a more general case where we are also given a positive integral edge weight function $\lambda$ (initially, we have $\lambda(e) = 1$ for every $e \in E$).
We may assume that $n > k+1$ (otherwise, we can compute all values $\delta(v)$ in ${\cal O}(k^2)$ time, if we use a brute-force algorithm).
We may further assume that all distances $D_S(v) = d(v,S)$, for every vertex $v$, are known throughout the entire procedure, which can always be ensured by running a BFS on $S$ before the algorithm starts.
Let $(T,{\cal X})$ be a tree decomposition of $G$ of width ${\cal O}(k)$. Such a tree decomposition can be computed in ${\cal O}(2^{{\cal O}(k)}n)$ time~\cite{Kor22}.
We often use that $m = {\cal O}(kn)$.
\begin{enumerate}
    \item We compute some bag ${\cal X}_t$, with $t \in V(T)$, such that every connected component of $G \setminus {\cal X}_t$ contains at most $n/2$ vertices. Such a bag always exists and it can be computed in ${\cal O}(kn)$ time if $(T,{\cal X})$ is given. For every $x \in {\cal X}_t \setminus S$, we can compute $\delta(x)$ in ${\cal O}(m) = {\cal O}(kn)$ time if we run a BFS on $x$.
    \item We group  the connected components of $G \setminus {\cal X}_t$ in two subsets $A$ and $B$ so that: $A \cup B = V$, $A \cap B = {\cal X}_t$ and $\max\{|A|,|B|\} \leq 2n/3$.
    We apply Lemma~\ref{lem:sep} twice in order to compute $\delta(a,B)$ for every $a \in A \setminus {\cal X}_t$ ($\delta(b,A)$, for every $b \in B \setminus {\cal X}_t$, resp.). This can be done in ${\cal O}(km+k2^kB(n,k)n) = {\cal O}(k2^kB(n,k)n)$ time.
    \item Let $G_A$ ($G_B$, resp.) be obtained from $G[A]$ ($G[B]$, resp.) by adding edges $xx'$ of weights $d(x,x')$, for every two $x,x' \in {\cal X}_t$.
    The computation of all distances $d(x,x')$ can be done in ${\cal O}(km) = {\cal O}(k^2n)$ time, if we run a BFS on every vertex of ${\cal X}_t$.
    We apply our algorithm recursively to $G_A,G_B$ in order to compute the values $\delta_A(a) = \sum\{\pi(a') \mid a' \in A, \ d(a,a') < D_S(a')\}$, for every $a \in A$, and the values $\delta_B(b) = \sum\{\pi(b') \mid b' \in B, \ d(b,b') < D_S(b')\}$, for every $b \in B$.
    \item For every $a \in A \setminus {\cal X}_t$, we set $\delta(a) = \delta(a,B) + \delta_A(a)$. In the same way, for every $b \in B \setminus {\cal X}_t$, we set $\delta(b) = \delta(b,A) + \delta_B(b)$.
\end{enumerate}
The recursion depth is in ${\cal O}(\log{n})$. Each recursive stage can be done in ${\cal O}(k2^kB(n,k)n)$ time (the running time is dominated by Step 2).
Therefore, the running time of the algorithm is in ${\cal O}(k2^{{\cal O}(k)}B(n,k)n\log{n})$.
\end{proof}

Finally, we give a one-to-many reduction from {\sc Balanced Vertex} with constant number of sites to Lemma~\ref{lem:wall}.

\begin{theorem}\label{thm:tw}
We can solve {\sc Balanced Vertex} in ${\cal O}(k2^{{\cal O}(k)}B(n,k)n|S|\log{n})$ time if $G$ has treewidth at most $k$.
\end{theorem}
\begin{proof}
For every $v \notin S$, our algorithm in what follows computes $\ell_\pi(v,S+v)$ and all values $\ell_\pi(s,S+v)$, for every $s \in S$.
Note that doing so, we can compute $L_\pi(S+v)$, for every $v \notin S$, in additional ${\cal O}(n|S|)$ time.
\begin{enumerate}
    \item We compute $\mathbf{Vor}(G,S)$. By Proposition~\ref{prop:compute-voronoi}, this can be done in ${\cal O}(m)$ time, that is in ${\cal O}(kn)$ if $G$ has treewidth $k$.
    Doing so, we can compute the loads $\ell_\pi(s,S)$, for every $s \in S$, in additional ${\cal O}(n)$ time.
    \item We apply Lemma~\ref{lem:wall} in order to compute $\ell_\pi(v,S+v)$, for every $v \notin S$.
    \item We then consider each $s \in S$ sequentially. We replace $S$ by the list $(s)$ with a unique site, and the cost function $\pi$ by a new cost function $\pi_s$ such that
    $$\begin{cases}
    \pi_s(u) = 0 \ \text{if} \ u \notin \mathbb{T}(s,S) \\
    \pi_s(u) = \pi(u) \ \text{otherwise.}
    \end{cases}$$
    Note that we can compute the new cost function $\pi_s$ in ${\cal O}(n)$ time, that is because we are given $\mathbf{Vor}(G,S)$.
    We apply Lemma~\ref{lem:wall} in order to compute, for every $v \notin S$, $\pi(W(v,s) \cap \mathbb{T}(s,S))$.
    Doing so, we can compute $\ell_\pi(s,S+v) = \ell_\pi(s,S) - \pi(W(v,s) \cap \mathbb{T}(s,S))$.
\end{enumerate}
Overall, since we apply Lemma~\ref{lem:wall} $|S|+1$ times, the running time is in ${\cal O}(k2^{{\cal O}(k)}B(n,k)n|S|\log{n})$.
\end{proof}

\begin{remark}\label{rk:tw}
Bringmann et al. proved that for every $\varepsilon > 0$, there exists a constant $c$ such that $B(n,k) = 2^{c k} n^\varepsilon$~\cite{BHM20}.
In particular, for every $\varepsilon > 0$, there exists a constant $c'$ such that the running time of our algorithm for bounded treewidth graphs is in ${\cal O}(2^{c'k}|S|n^{1+\varepsilon})$.
The exponential dependency on the treewidth is necessary, even for one site, if one assumes the Hitting Set Conjecture. Indeed, the graph constructed in the proof of Theorem~\ref{thm:hs-hard} has treewidth ${\cal O}(\log{n})$. 
\end{remark}

\subsection{Proper interval graphs}\label{sec:interval}

We devote our last and arguably most technical subsection to proper interval graphs.
Recall that an interval graph is the intersection graph of a family of intervals on the real line.
A realization of an interval graph $G$ is a mapping of its vertices to closed intervals on the real line, so that two vertices are adjacent in $G$ if and only if their corresponding intervals are intersecting.
It is a {\em proper realization} if there are no two vertices such that the interval of one is contained in the interval of the other.
A {\em proper interval graph} is one admitting a proper realization.

Our scheme only applies to proper interval graphs, not to interval graphs, and it is good to explain why it is so.
Roughly, we would like to mimic the scheme for paths (see Lemma~\ref{lem:path}).
A natural path representation for the interval graphs is to order the vertices by nondecreasing left boundary, within some fixed realization.
Then, what we would like to prove is the existence of some constant $\kappa$ (ideally, $\kappa = 1$) such that if we insert a new site $v$ at the end of $S$, then only the $\kappa$ closest sites to $v$ on the left and right of the path representation could have their respective territories modified.
Unfortunately, this is impossible due to the potential existence of vertices whose interval covers an arbitrary number of other intervals.
Therefore, we restrict ourselves to proper realizations, for which the latter case cannot happen.
Some important properties of proper realizations are captured by the following distance formula:

\begin{lemma}[\cite{GaP08}]\label{lem:pties-proper-int}
Let $G=(V,E)$ be a proper interval graph, and let $x_0$ be the vertex with the minimum left boundary in some proper realization of $G$.
We can compute in ${\cal O}(n)$ time a total ordering $\sigma : V \to \{1,\ldots,n\}$ so that, for every vertices $u$ and $v$ with $d(u,x_0) \leq d(v,x_0)$:
$$d(u,v) = \begin{cases}
d(v,x_0) - d(u,x_0) \ \text{if} \ d(u,x_0) < d(v,x_0) \ \text{and} \ \sigma(v) < \sigma(u) \\
d(v,x_0) - d(u,x_0) + 1 \ \text{otherwise}.
\end{cases}$$
\end{lemma}

First, we combine Lemma~\ref{lem:pties-proper-int} with orthogonal range queries (see Sec.~\ref{sec:tw}) in order to precompute the load of every new site in a prioritized Voronoi diagram. 
We could easily do so with $3$-dimensional range queries because, for every vertex $u$, we only need to memorize $\sigma(u)$ and $d(u,x_0),d(u,S)$.
However, we can shave some polylogarithmic factors by reducing further down to $1$-dimensional range queries.
Indeed, a $1$-range tree is just a balanced binary search tree\footnote{Values are stored at the leaves, but this feature is not important in our case.}.
In particular, by using a self-balanced binary search tree implementation, we can serve range queries on a dynamic collection of values. 
Doing so, we only need to store $\sigma(u)$, for every vertex $u$, while other relevant distance information can be deduced contextually for the vertices $u$ that are currently stored in the search tree.

\begin{lemma}\label{lem:compute-territory}
Let $G=(V,E)$ be a proper interval graph, let $x_0$ be the vertex with the minimum left boundary in some proper realization of $G$, and let $\sigma$ be as in Lemma~\ref{lem:pties-proper-int}.
For every nonnegative cost function $\pi$, and every proper vertex list $S$, we can compute the values $\ell_\pi(v,S+v)$, for every $v \notin S$, in ${\cal O}(m+n\log{n})$ time.
\end{lemma}
\begin{proof}
We first describe the algorithm.
\begin{enumerate}
    \item We compute all distances $d(v,x_0), d(v,S)$, for every vertex $v$.
    Doing so, we partition $V$ in distance layers $X_0=\{x_0\}, X_1 = N(x_0), X_2, \ldots, X_{e(x_0)}$, where $e(x_0) = \max_{v \in V}d(x_0,v)$ and, for every $j$ with $0 \leq j \leq e(x_0)$, $X_j = \{v \in V \mid d(x_0,v) = j \}$.
    This can be done in ${\cal O}(n+m)$ time.
    Let us further compute $\mathbf{Vor}(G,S)$, which by Proposition~\ref{prop:compute-voronoi} can also be done in ${\cal O}(n+m)$ time.
    For convenience, for every vertex $v$, we denote in what follows by $s_v$ the unique site of $S$ such that $v \in \mathbb{T}(s_v,S)$.
    
    \item We consider each distance layer $X_0,X_1,\ldots, X_{e(x_0)}$ sequentially. Throughout all iterations of this phase, we maintain a counter $\Lambda^-$ and a balanced binary search tree $T^-$ such that the following invariants hold before considering any layer $X_j$:
    \begin{itemize}
        \item $\Lambda^- = \sum\{ \pi(u) \mid d(x_0,u) < j, \ \text{and} \ d(u,s_u) > j - d(x_0,u) + 1  \}$.
        
        \smallskip
        Since the cost $\pi(u)$ of every vertex $u$ needs to be added and removed from $\Lambda^-$ at most once, we can maintain the counter $\Lambda^-$ throughout this whole phase in total ${\cal O}(n)$ time. 
        \item we store in $T^-$ all vertices $u$ such that: $d(x_0,u) < j$, and $d(u,s_u) = j - d(x_0,s_u) + 1$. A vertex $u$ is identified in $T^-$ with the value $\sigma(u)$. Furthermore, at each vertex $u$ in $T^-$ we store the sum of the costs $\pi(u')$ of all vertices $u'$ in its rooted subtree.
        
        Each vertex $u$ is inserted and deleted from $T^-$ at most once. Standard implementations such as AVL~\cite{AVL62} can be easily modified in order to dynamically maintain the cost of rooted subtrees after each insertion/deletion. As a result, we can maintain $T^-$ throughout this whole phase in total ${\cal O}(n\log{n})$ time. 
    \end{itemize}
    For every $v \in X_j$, we compute the cumulative cost of all vertices $u \in T^-$ such that $\sigma(v) < \sigma(u)$. Since every node of $T^-$ stores the cost of its rooted subtree, this can be done in ${\cal O}(\log{n})$ time by using range searching techniques, namely: we search for the vertex $u$ of $T^-$ such that $\sigma(u) > \sigma(v)$ and $\sigma(u)$ is minimized, then we process the costs of all vertices on the path to the root and of their respective right subtrees (details left out can be found, {\it e.g.}, in~\cite{Ben79,BHM20}). Let $\lambda^-(v) = \Lambda^- + \sum\{ \pi(u) \mid u \in T^-, \sigma(v) < \sigma(u) \}$. 
    
    \item Roughly, we do the reverse of the previous phase. Specifically, we consider each distance layer $X_{e(x_0)},X_{e(x_0)-1},\ldots,X_0$ sequentially. Throughout all iterations of this phase, we maintain a counter $\Lambda^+$ and a balanced binary search tree $T^+$ such that the following invariants hold before considering any layer $X_j$:
    \begin{itemize}
        \item $\Lambda^+ = \sum\{ \pi(u) \mid d(x_0,u) > j, \ \text{and} \ d(u,s_u) > d(x_0,u) - j + 1  \}$.
    
        \item we store in $T^+$ all vertices $u$ such that: $d(x_0,u) > j$, and $d(u,s_u) = d(x_0,s_u)-j + 1$. A vertex $u$ is identified in $T^+$ with the value $\sigma(u)$. Furthermore, at each vertex $u$ in $T^+$ we store the sum of the costs $\pi(u')$ of all vertices $u'$ in its rooted subtree.
    \end{itemize}
    For every $v \in X_j$, we compute the cumulative cost of all vertices $u \in T^+$ such that $\sigma(v) > \sigma(u)$. Let $\lambda^+(v) = \Lambda^+ + \sum\{ \pi(u) \mid u \in T^+, \sigma(v) > \sigma(u) \}$. 
\end{enumerate}
The total running time is in ${\cal O}(m+n\log{n})$.

\medskip
Let $j$ with $0 \leq j \leq e(x_0)$ be fixed.
We set $\Pi(j) = \sum\{ \pi(u) \mid u \in X_j, \ d(u,s_u) > 1 \}$.
For every $v \in X_j \setminus S$, we claim that we have
$$\ell_\pi(v,S+v) = \begin{cases}
\lambda^-(v) + \lambda^+(v) + \Pi(j) \ \text{if} \ d(v,s_v) > 1 \\
\lambda^-(v) + \lambda^+(v) + \Pi(j) + \pi(v) \ \text{else}
\end{cases}$$
Indeed, let $u \in V \setminus \{v\}$ be arbitrary. 
There are three different cases:
\begin{itemize}
    \item If $d(u,x_0) = j$, then by Lemma~\ref{lem:pties-proper-int}, $d(u,v) = 1$. In particular, $u \in \mathbb{T}(v,S+v)$ if and only if $d(u,s_u) > 1$. Therefore, $\ell_\pi(v,S+v) \geq \Pi(j)$, and if $d(v,s_v) =1$ then $\ell_\pi(v,S+v) \geq \Pi(j)+\pi(v)$. 
    \item If $d(u,x_0) < j$, then by Lemma~\ref{lem:pties-proper-int}, $d(u,v) \in \{j-d(u,x_0),j-d(u,x_0)+1\}$. In particular, $u \notin \mathbb{T}(v,S+v)$ if $d(u,s_u) \leq j-d(u,x_0)$. Furthermore, $u \in \mathbb{T}(v,S+v)$ if $d(u,s_u) > j-d(u,x_0)+1$. Otherwise, $d(u,s_u) = j-d(u,x_0)+1$, and so we have $u \in \mathbb{T}(v,S+v)$ if and only if $d(u,v) = j-d(u,x_0)$, which by Lemma~\ref{lem:pties-proper-int} is equivalent to $\sigma(v) < \sigma(u)$. As a result, $\ell_\pi(v,S+v) \geq \lambda^-(v)$. 
    \item Otherwise, $d(u,x_0) > j$. By Lemma~\ref{lem:pties-proper-int}, $d(u,v) \in \{d(u,x_0)-j,d(u,x_0)-j+1\}$. In particular, $u \notin \mathbb{T}(v,S+v)$ if $d(u,s_u) \leq d(u,x_0)-j$. Furthermore, $u \in \mathbb{T}(v,S+v)$ if $d(u,s_u) > d(u,x_0)-j+1$. Otherwise, $d(u,s_u) = d(u,x_0)-j+1$, and so we have $u \in \mathbb{T}(v,S+v)$ if and only if $d(u,v) = d(u,x_0)-j$, which by Lemma~\ref{lem:pties-proper-int} is equivalent to $\sigma(u) < \sigma(v)$. As a result, $\ell_\pi(v,S+v) \geq \lambda^+(v)$. 
\end{itemize}
Our claim directly follows from this case analysis.
\end{proof}

From another application of Lemma~\ref{lem:pties-proper-int}, where we process distance layers from $x_0$ in a different way than what we did for Lemma~\ref{lem:compute-territory}, we obtain the following complementary result:

\begin{lemma}\label{lem:process-clique}
Let $G=(V,E)$ be a proper interval graph, let $x_0$ be the vertex with the minimum left boundary in some proper realization of $G$, and let $\sigma$ be as in Lemma~\ref{lem:pties-proper-int}.
For every nonnegative cost function $\pi$, and every proper vertex list $S$, we can compute the values $\Lambda(v) = \max\{\ell_\pi(s,S+v) \mid s \in S\}$, for every vertex $v \notin S$, in ${\cal O}(m+n\log{n})$ time.
\end{lemma}
\begin{proof}
We need some notations for what follows.
Let $e(x_0)  = \max_{v \in V}d(x_0,v)$.
For every $j$ with $0 \leq j \leq e(x_0)$, let $X_j = \{ v \in V \mid d(x_0,v) = j \}$.
Finally, for every vertex $v$, let $s_v$ be the unique site of $S$ such that $v \in \mathbb{T}(s_v,S)$.

Roughly, the algorithm consists of two phases.
During its first phase, we compute the Voronoi diagram of $G$ with respect to $S$, and we create a max-heap ${\cal H}$ storing all sites in $S$.
Then, during the second and main phase, we process the heap in order to compute $\Lambda(v)$, for every $v \in V \setminus S$.
More precisely, we iterate over the distance layers $X_0,X_1,\ldots,X_{e(x_0)}$ and, for every $j$ with $0 \leq j \leq e(x_0)$, we compute all values $\Lambda(v)$, for $v \in X_j \setminus S$, sequentially.
Intuitively, when we consider a vertex $v \in X_j \setminus S$, the key of every site $s$ in ${\cal H}$ must be equal to $\ell_\pi(s,S+v)$.

\bigskip
\noindent
{\bf First Phase.}
We compute $\mathbf{Vor}(G,S)$, which by Proposition~\ref{prop:compute-voronoi} can be done in ${\cal O}(n+m)$ time.
We compute, also in ${\cal O}(n+m)$ time, the distances $d(x_0,v), d(v,S)$ for every vertex $v$.
Doing so, we can partition $V$ in $X_0,X_1,\ldots,X_{e(x_0)}$.
Finally, we create a max-heap ${\cal H}$ and, for every site $s$ of $S$, we insert $s$ into ${\cal H}$ with key equal to 
$$key(s) = \ell_\pi(s,S)$$
Being given $\mathbf{Vor}(G,S)$, we can compute all the keys in ${\cal O}(n)$ time.
Every insertion in a heap can be done in ${\cal O}(\log{n})$ time.
Therefore, we can initialize the heap ${\cal H}$ in ${\cal O}(n\log{n})$ time.

\smallskip
Throughout every loop $j$ of the second phase, we maintain the following invariant for ${\cal H}$: for every site $s$ of $S$,
\[key(s) = \ell_\pi(s,S) - \sum\{ \pi(u) \mid u \in \mathbb{T}(s,S) \ \text{and} \ d(u,s_u) > |j-d(u,x_0)| +1 \}\]
Intuitively, this is because any such vertex $u$ {\em must} be in the territory of $v \notin S$, provided $v \in X_j$. 
For that, before we start the $j^{th}$ loop, we need to actualize the keys of all sites.
More precisely, for every vertex $u$:
\begin{itemize}
    \item We set $key(s_u) = key(s_u)-\pi(u)$ if the following conditions hold:
    $$\begin{cases}
        d(u,s_u) \geq 2 \\
        j = \max\left\{0, d(u,x_0) - d(u,s_u) + 2\right\}
    \end{cases}$$
    Indeed, if $d(u,s_u) \leq 1$, then there is no $j'$ such that $d(u,s_u) > |j'-d(u,x_0)| + 1$.
    Otherwise, for $j \leq d(u,x_0)$, we have $d(u,s_u) > d(u,x_0) - j + 1$ if and only if $j > d(u,x_0)-d(u,s_u)+1$.
    \item We set $key(s_u) = key(s_u)+\pi(u)$ if the following conditions hold:
    $$\begin{cases}
        d(u,s_u) \geq 2 \\
        j = d(u,s_u) + d(u,x_0) - 1
    \end{cases}$$
    Indeed, if $d(u,s_u) \geq 2$ and $j \geq d(u,x_0)$, then we have $d(u,s_u) > j - d(u,x_0) + 1$ if and only if $j < d(u,s_u) + d(u,x_0) - 1$.
\end{itemize}
For every vertex $u$, there are at most two loops such that we modify $key(s_u)$.
Therefore, we need to modify any key ${\cal O}(n)$ times during the second phase.
Overall, we can maintain our invariant throughout the second phase in ${\cal O}(n\log{n})$ time.

\bigskip
\noindent
{\bf Second Phase.}
We are now ready to describe the main loop of the algorithm.
We consider each distance layer $X_0,X_1,\ldots,X_{e(x_0)}$ sequentially.
Let $j$ with $0 \leq j \leq e(x_0)$ be fixed.
We compute $\Lambda(v)$ for every vertex $v \in X_j \setminus S$.

For that, we define the {\em frontier} $F_j$, as follows:
$$F_j = \{ u \in V \setminus X_j \mid d(u,s_u) = |j-d(u,x_0)| + 1\}.$$ 
Intuitively, being given a vertex $v \in X_j \setminus S$, the vertices $u$ in the frontier may or may not be in $\mathbb{T}(v,S+v)$ depending on whether $\sigma(u)$ is bigger or smaller than $\sigma(v)$. 
Determining, for every site $s$, the vertices in $F_j \cap \mathbb{T}(s,S)$ that belong to $v$'s territory, is the main difficulty in our algorithm.
In what follows, we find it convenient to bipartition $F_j$ in $F_j^<, F_j^>$ such that:
\begin{align*}
    F_j^< &= \{ u \in F_j \mid d(x_0,u) < j \} \\
    F_j^> &= \{ u \in F_j \mid d(x_0,u) > j \}
\end{align*}
We subdivide the loop in many intermediate steps:
\begin{enumerate}
    \item For every site $s$ of $S$, we set: $$key(s) = key(s) - \sum\{\pi(u) \mid u \in F_j^< \cap \mathbb{T}(s,S)\} $$
    Since we are given $\mathbf{Vor}(G,S)$, this can be done in ${\cal O}(|F_j^<|\log{n})$ time by scanning the vertices of $F_j^<$. 
    \item We order the vertices of $F_j \cup X_j$ according to $\sigma$. This can be done in ${\cal O}(|F_j|\log{n}+|X_j|\log{n})$ time by using any fast sorting algorithm.
    \item We scan the vertices of $F_j \cup X_j$ by increasing $\sigma$ values:
    \begin{itemize}
        \item {\it Case of a vertex $u \in F_j^<$}. We set $key(s_u) = key(s_u) + \pi(u)$.
        \item {\it Case of a vertex $u \in F_j^>$}. We set $key(s_u) = key(s_u) - \pi(u)$.
        \item {\it Case of a vertex $v \in X_j \setminus S$}. First, if $d(s_v,v) = 1$, then we set $key(s_v) = key(s_v) - \pi(v)$.
        We set $\Lambda(v)$ to be the maximum key in ${\cal H}$.
        Finally, if $d(s_v,v) = 1$, then we reset $key(s_v) = key(s_v) + \pi(v)$.
    \end{itemize}
    This can be done in ${\cal O}(\log{n})$ time per vertex considered.
    \item We reset the keys in ${\cal H}$, namely, for every site $s$ of $S$, we set: $$key(s) = key(s) + \sum\{\pi(u) \mid u \in F_j^> \cap \mathbb{T}(s,S)\} $$
    It takes ${\cal O}(|F_j^>|\log{n})$ time.
\end{enumerate}
The total running time for the $j^{th}$ loop is in ${\cal O}(|F_j|\log{n}+|X_j|\log{n})$.
Since every vertex $u$ can appear in the frontier $F_{j'}$ for at most two indices $j'$, the overall running time of the second phase is in ${\cal O}(n \log{n})$.

\medskip
In order to prove correctness of our algorithm, let $j$ with $0 \leq j \leq e(x_0)$ be fixed.
Let $v \in X_j \setminus S$ be arbitrary.
In order to prove that we correctly computed $\Lambda(v)$, it suffices to prove that at the time we consider this vertex during the $j^{th}$ loop we have $key(s) = \ell_\pi(s,S+v)$ for every site $s$ of $S$.
That is indeed the case because:
\begin{align*}
    key(s) &= \ell_\pi(s,S) - \sum\{ \pi(u) \mid u \in \mathbb{T}(s,S) \ \text{and} \ d(u,s_u) > |j-d(u,x_0)| +1 \} \\
    &\hspace{50pt}- \sum\{\pi(u) \mid u \in F_j^< \cap \mathbb{T}(s,S)\} \\
    &\hspace{50pt}+ \sum\{\pi(u) \mid u \in F_j^< \cap \mathbb{T}(s,S) \ \text{and} \ \sigma(u) < \sigma(v)\} \\
    &\hspace{50pt}- \sum\{\pi(u) \mid u \in F_j^> \cap \mathbb{T}(s,S) \ \text{and} \ \sigma(u) < \sigma(v)\} \\
    &\hspace{50pt}-\mathbf{1}[s = s_v \ \text{and} \ d(s,v)=1] \cdot \pi(v) \\
    &= \ell_\pi(s,S) - \sum\{ \pi(u) \mid u \in \mathbb{T}(s,S) \setminus \{v\} \ \text{and} \ d(u,s_u) > |j-d(u,x_0)| +1 \} \\
    &\hspace{50pt}- \sum\{\pi(u) \mid u \in F_j^< \cap \mathbb{T}(s,S) \ \text{and} \ \sigma(u) > \sigma(v) \} \\
    &\hspace{50pt}- \sum\{\pi(u) \mid u \in F_j^> \cap \mathbb{T}(s,S) \ \text{and} \ \sigma(u) < \sigma(v)\} \\
    &\hspace{50pt}-\mathbf{1}[s = s_v] \cdot \pi(v) \\
     &= \ell_\pi(s,S) - \sum\{ \pi(u) \mid u \in \mathbb{T}(s,S) \setminus \{v\} \ \text{and} \ d(u,s_u) > |j-d(u,x_0)| +1 \} \\
    &\hspace{50pt}- \sum\{\pi(u) \mid u \in F_j^< \cap \mathbb{T}(s,S) \ \text{and} \ d(u,v) = j-d(u,x_0) \} \\
    &\hspace{50pt}- \sum\{\pi(u) \mid u \in F_j^> \cap \mathbb{T}(s,S) \ \text{and} \ d(u,v) = d(u,x_0)-j\} \\
    &\hspace{50pt}-\mathbf{1}[s = s_v] \cdot \pi(v) \\
    &= \ell_\pi(s,S) - \sum\{ \pi(u) \mid u \in \mathbb{T}(s,S) \ \text{and} \ d(u,v) < d(u,s) \} \\
    &= \ell_\pi(s,S+v)
\end{align*}
\end{proof}

The following theorem summarizes our results in this subsection:

\begin{theorem}\label{thm:proper-int}
We can solve {\sc Balanced Vertex} in ${\cal O}(m+n\log{n})$ time if $G$ is a proper interval graph.
\end{theorem}
\begin{proof}
We compute a proper realizer, which can be done in ${\cal O}(n+m)$ time~\cite{Cor04}.
Then, in additional ${\cal O}(n)$ time, we compute the vertex $x_0$ of minimum left boundary and the ordering $\sigma$ of Lemma~\ref{lem:pties-proper-int}.
We are done applying Lemmas~\ref{lem:compute-territory} and~\ref{lem:process-clique}.
\end{proof}

\section{Conclusion}\label{sec:ccl}
In this paper, we introduced a new problem on graph Voronoi diagrams, and we completely settled its complexity for general graphs under plausible complexity assumptions.
We left open whether our conditional time lower bound could also hold under weaker complexity hypotheses.

We proposed various ways to tackle with our new problem on tree-like and path-like topologies.
In this respect, two immediate questions following our work could be whether our results can be extended to bounded-treewidth graphs with nonconstant number of sites, and to interval graphs.

\medskip
\noindent\textbf{Acknowledgement}\textit{.} This work was supported by a grant of the Ministry of Research, Innovation and Digitalization, CCCDI - UEFISCDI, project number PN-III-P2-2.1-PED-2021-2142, within PNCDI III.

\bibliographystyle{abbrv}
\bibliography{biblio}

\end{document}